\documentclass[journal]{IEEEtran}
\IEEEoverridecommandlockouts
\usepackage{xcolor}
\usepackage{amssymb}
\usepackage{diagbox}
\usepackage{amsmath}
\usepackage{adjustbox}
\usepackage{mleftright}
\usepackage{caption}
\usepackage{cite}
\usepackage[caption=false,font=footnotesize]{subfig}
\usepackage[english]{babel}
\usepackage{mathrsfs}
\usepackage{booktabs}

\newtheorem{corollary}{Corollary}

\newtheorem{theorem}{Theorem}
\newtheorem{lemma}{Lemma}
\newtheorem{remark}{Remark}

\newenvironment{proof}{{\indent \it {Proof}:\quad}}{\hfill $\square$\par}
\newtheorem{definition}{Definition}
\usepackage{amsmath}
\usepackage{graphicx}
\usepackage[colorlinks=true, allcolors=blue]{hyperref}
\usepackage{color}
\usepackage{lettrine}
 
\usepackage{amssymb}
\usepackage{array}

\title{Semi-Blind Fluid Antenna System: Port Selection via Statistical Analysis}

\author{Tianyu Han, Yongxu Zhu, Kai-Kit Wong, Gan Zheng, Chan-Byoung Chae, and Xiaohu You
\vspace{-6mm}
\thanks{This work has been accepted for publication in \emph{IEEE Transactions on Wireless Communications} under the title ``Partial Fluid Antenna System: Port Selection via Statistical Analysis,'' DOI: 10.1109/TWC.2026.3722171. \textcolor{blue}{The terms ``partial'' and ``semi-blind'' refer to the same concept. Readers should not be confused by this difference in terminology. We encourage readers to use the term ``semi-blind FAS'' when citing and referring this work, as it more appropriately describes the proposed approach.}}
\thanks{T. Han and G. Zheng are with the School of Engineering, University of Warwick, Coventry, UK (e-mail: \{tianyu.han, gan.zheng\}@warwick.ac.uk).}
\thanks{Y. Zhu and X. You are with the National Communications Research Laboratory, Southeast University, Nanjing, China (e-mail: \{yongxu.zhu, xhyu\}@seu.edu.cn).}
\thanks{K. Wong is affiliated with the Department of Electronic and Electrical Engineering, University College London, Torrington Place, UK and also with the Yonsei Frontier Lab., Yonsei University, Seoul 03722, South Korea (e-mail: kai-kit.wong@ucl.ac.uk)}
\thanks{C.-B. Chae is with the School of Integrated Technology, Yonsei University, Seoul, 03722 South Korea (e-mail: cbchae@yonsei.ac.kr).}
\thanks{Corresponding author: Yongxu Zhu}
}

\begin{document}
\maketitle

\begin{abstract}
The fluid antenna system (FAS) enables position reconfigurability, granting the transceiver access to a high-resolution spatial signal. A potential drawback of real-time FAS, however, is that it requires complete channel state information (CSI) for each FAS port at every communication time slot, an approach referred to as ideal-FAS. Recognizing the difficulties of achieving ideal-FAS, we propose a FAS scheme based on incomplete CSI, referred to as semi-blind FAS. This paper first introduces the spatial-temporal framework of FAS, upon which the proposed semi-blind FAS is developed. The proposed semi-blind FAS is lightweight and computationally efficient, scalable to an arbitrary number of ports and time slots, and operates without pre-training or deep learning structures. The scheme effectively exploits incomplete historical CSI to estimate the conditional distribution across all FAS ports at the desired time slot, thereby identifying the statistical optimal port for signal reception. Generally, the key idea of semi-blind FAS is to select the optimal port through conditional distribution analysis, from a statistical perspective, with optimality defined according to the scenario of interest. Moreover, we derive a closed-form expression for the placement of optimal port, where optimality is defined as the port that minimizes the outage probability, in the special case where only a single port CSI is available. Inspired by information-theoretic entropy, we further develop the residual entropy power ratio to characterize how physical parameters influence the performance gap between semi-blind FAS and ideal-FAS. Our analysis reveals that estimation performance depends not only on the number of sampled ports and time slots, but also on the specific indices of ports with given CSI at each time slot, i.e., the port sampling strategy. This critical factor has been largely overlooked in existing port estimation studies. In addition, we establish the Markov condition for the proposed semi-blind FAS, which provides insights into the physical design of time slot duration and the number of historical CSI samples required for semi-blind FAS to closely approximate ideal-FAS. Numerical results demonstrate that the proposed semi-blind FAS achieves performance comparable to, and in some cases indistinguishable from, that of ideal-FAS, while requiring significantly fewer port CSI measurements and lower port switching speeds.
\end{abstract}
	
	
\section{Introduction}
\subsection{Motivation}
\lettrine{M}{ultiple}-input multiple-output (MIMO) technology has been recognized as a cornerstone for enhancing wireless communication performance \cite{telatar1999capacity}. Recent advances have demonstrated that massive MIMO, which equips base stations (BSs) with hundreds of antennas, can further amplify system performance, most notably by improving spectral and energy efficiency, suppressing interference, and realizing channel hardening effects \cite{marzetta2016fundamentals}. However, these gains come at a cost. The deployment of large antenna arrays entails significant capital expenditures, including the cost of numerous antennas and associated radio frequency (RF) chains. Moreover, the computational burden and power consumption of signal processing over such large-scale systems introduces considerable operational complexity. These challenges, particularly pilot contamination, hardware impairments, and channel estimation inaccuracies, pose significant obstacles to the widespread commercial deployment of massive MIMO \cite{lu2014overview}. It is therefore desirable to develop solutions that can achieve performance comparable to massive MIMO, but without incurring its deployment and operational costs.

\subsection{FAS and Literature Review}
Motivated by the aforementioned challenges, a novel antenna architecture known as the fluid antenna system (FAS) has recently been proposed. FAS broadly encompasses any software-controlled fluidic, dielectric, or conductive structures, such as liquid-based antennas, pixel-based antennas, and metasurfaces that are capable of dynamically reconfiguring their shape, size, position, length, orientation, and other radiation characteristics. This reconfigurability enables the realization of spatial diversity within a physically constrained environment, and the concept of FAS was first introduced to the wireless communication domain by Wong \textit{et al.} in \cite{Fluid_antenna_system}. The fundamental concept of FAS involves the dense deployment of a large number of antenna ports within a limited physical space, while utilizing only a single RF chain \cite{Lee_fluid_Antenna}. Although this single-chain architecture imposes constraints on signal processing capabilities, existing studies have demonstrated that, with effective port selection strategies, FAS can achieve performance that is comparable to, and in some cases better than, conventional techniques that rely on multiple RF chains, such as maximum ratio combining (MRC) \cite{Fluid_antenna_system}. Since then, numerous studies have explored the integration of FAS with other advanced technologies. In \cite{lai2024fas}, the authors provided the performance analysis of FAS in conjunction with the reconfigurable intelligent surface (RIS). The work in \cite{yao2024fas} integrated cognitive radio (CR) with FAS to enhance system performance. \textcolor{black}{Recent experimental developments in FAS were presented in \cite{lu2025fluid}, while a comprehensive tutorial on FAS was provided in \cite{new2024tutorial}.}

FAS can also support multiple access by allowing multiple users to be served simultaneously on the same time-frequency resource, a concept referred to as fluid antenna multiple access (FAMA). In FAMA, the FAS selects the `best' port among all available ports, thereby improving the received signal-to-interference ratio (SIR). In \cite{fast_FAMA}, it was proposed that FAS ports can be switched at the symbol level, a method termed fast FAMA ($f$-FAMA). \textcolor{black}{Later, \cite{slow_FAMA} introduced slow FAMA ($s$-FAMA), where port switching occurs only when the fading channel changes.} \textcolor{black}{A recent comprehensive survey of FAMA is given in \cite{shah2024survey}.} 

It can be concluded that the primary difference between massive MIMO and FAS lies in the number of RF chains required. The latter requires significantly fewer RF chains, typically only a single RF chain per FAS, which consequently reduces both the capital and operational expenditures. However, this reduction in cost comes with certain trade-offs. Given the presence of only a single RF chain, the most straightforward method of CSI estimation in FAS must be performed sequentially, port by port. Meanwhile, the core requirement underlying all FAS related work is the perfect CSI for all FAS ports at each communication time slot. It appears to be an inherent and perhaps inescapable trade-off. To obtain perfect CSI for all ports simultaneously, multiple RF chains are required to enable concurrent CSI estimation. However, the inclusion of numerous RF chains significantly increases hardware cost and system complexity. \textcolor{black}{A similar strategy for reducing the number of RF chains while preserving the beamforming gain of massive MIMO is adopted in hybrid analog--digital beamforming \cite{el2014spatially}. This hybrid architectures face a similar challenge of estimating a large number of wireless channels using only a limited number of RF chains.}

\textcolor{black}{Recent studies in the FAS literature sought to address this challenge by leveraging deep learning techniques.} In \cite{chai2022port}, several deep learning methods were proposed to reconstruct the full port CSI based on the CSI observed from a limited subset of ports in FAS. Recently, the authors in \cite{zou2023online} introduced an online learning framework to tackle the port selection challenge in dynamically changing channel conditions. Additionally, the works in \cite{waqar2023deep} investigated the CSI estimation problem in $s$-FAMA scenarios using deep learning methods. \textcolor{black}{In \cite{10447499}, correlation-specific sub-networks are proposed to reconstruct the full FAS channel. The authors in \cite{11180048} investigate the amount of training required to reconstruct the CSI of all FAS ports.} \textcolor{black}{For channel estimation in hybrid analog–digital beamforming systems, many existing works exploit the sparse nature of millimeter-wave channels to address this issue \cite{alkhateeb2014channel, venugopal2017channel}. Only a limited number of studies consider channel estimation without relying on the sparsity assumption of the wireless channel \cite{bogale2015hybrid}, however, this approach inevitably leads to the non-convex optimization that in turn increase the computational complexity.}

\textcolor{black}{In summary, although the aforementioned approaches have demonstrated effectiveness in addressing the channel estimation issue in the literature of FAS and  hybrid analog–digital beamforming, they still face considerable challenges. \textcolor{black}{Specifically, in the context of FAS, the adoption of deep learning inevitably introduces substantial overhead in terms of training cost, parameter storage, and computational complexity. This leads to the implementation of FAS being constrained by the availability of advanced processing hardware.} In the domain of hybrid analog–digital beamforming, \textcolor{black}{most methods rely on channel sparsity assumptions.} Approaches that relax this assumption typically incur additional computational burden.} Recall that, it is our goal to achieve performance comparable to massive MIMO, but without incurring its operation complexity and development cost. 

\subsection{Contributions}
In summary, existing works typically treat the FAS port estimation problem as a prediction task within a deep learning framework, which might be a redundant undertaking. In this paper, we revisit this problem from a performance analysis perspective. This conceptual shift simplifies the problem formulation and avoids reliance on complex deep learning architectures. Specifically, rather than estimating the CSI of each individual FAS port, we analyze the conditional distribution associated with each port and thereby identify the statistically optimal port for signal reception. Once the optimal port has been identified, classical channel estimation methods can be applied to obtain the CSI of that port. {\color{black}The key idea of semi-blind FAS is summarized as follows:
\begin{center}
\textit{Identify the statistically optimal port through conditional distributional analysis based on incomplete historical CSI, rather than recovering the exact CSI values across all ports.}
\end{center}
It should be pointed out that semi-blind FAS goes beyond mere port selection and is better understood as:
\begin{center}
\textit{An FAS framework in which the system operates with incomplete CSI.}
\end{center}}

We refer to the case where perfect CSI is available for each FAS port at each time slot as ideal-FAS. The port selection criteria of ideal-FAS can be expressed as
\begin{align}
    k_{\mathrm{ideal}\text{-}\mathrm{FAS}}^{t}  = \arg\max_{k} |h_{k}^{t}|,
\end{align}
which requires the perfect CSI of each FAS port at each communication time slot. In contrast, semi-blind FAS only requires CSI \textcolor{black}{over} a sub-set of ports, $\bar{\boldsymbol{\mathrm{h}}}$, and find the \textcolor{black}{optimal} port that satisfy 
\begin{align}
    k_{\mathrm{\textcolor{black}{semi}}\text{-}\mathrm{FAS}}^{t} = \arg\min_{k} \mathcal{OP}_{|h_{k}^{t}|{\big |}\bar{\boldsymbol{\mathrm{h}}}=\boldsymbol{a}_{c}}.
\end{align}
\textcolor{black}{Here, the optimal port is defined as the one that minimizes the outage probability. However, this criterion is not unique and may vary depending on the specific task under consideration, as will be detailed in Section~\ref{subsec:semi-blind-FAS}.}

In summary, we have made the following contributions:
\begin{itemize}
    \item This paper first introduces the spatial-temporal structure of FAS, upon which the proposed semi-blind FAS is developed. Different from ideal-FAS, which requires full CSI over the full sets of ports at each time slot, semi-blind FAS finds the \textcolor{black}{statistical} optimal port only when a sub-set of ports' historical CSI is available. The proposed semi-blind FAS offers a computation efficient solution that operates without the need for pretraining. It is universally adaptable to any number of time slots and any number of ports with available CSI.
    \item Conditioned on the sub-set of ports with available historical CSI, we evaluate the distribution of each FAS port at the desired time slot to compute its outage probability. \textcolor{black}{Accordingly, the optimal port is defined as the one that minimizes the outage probability. This definition is not unique; more generally, semi-blind FAS selects ports based on statistical analysis, and the definition of optimality may vary with the adopted performance criterion.}
    \item \textcolor{black}{We also consider the special case where only the CSI of a single port is available, for which closed-form expressions for both the optimal correlation coefficient and the placement of the optimal port are derived, with optimality defined in terms of minimizing the outage probability.}
    \item Inspired by information theory, we propose a prior evaluation metric, residual entropy power ratio, as a measure of the performance of the proposed semi-blind FAS. Moreover, according to the residual entropy power ratio, it is found that the performance of FAS port estimation, is not solely determined by the number or percentage of ports with available CSI, rather, critically depends on their specific port index. Therefore, an accurate characterization of the estimation performance requires the use of residual entropy power ratio, rather than merely stating the number or percentage of ports with available CSI.
    \item The Markov condition of semi-blind FAS is also provided, which provide useful insights into the trade-off between the length of the time interval and the number of historical CSI required.
\end{itemize}

The remainder of the paper is organized as follows. In Section~\ref{section:System_Model}, we introduce the spatial-temporal framework of FAS. Then Section~\ref{section:partial_FAS} provides the proposed semi-blind FAS structure. Some physical insights of the proposed semi-blind FAS is provided in Section~\ref{section:Physical_Insights}. In Section~\ref{section:results}, numerical results are presented and finally, we
provide some concluding remarks in Section~\ref{section:Conclusion}.
	
\section{System Model}\label{section:System_Model}
\begin{figure*}[!htbp]
    \centering
    \includegraphics[width=1.0\linewidth]{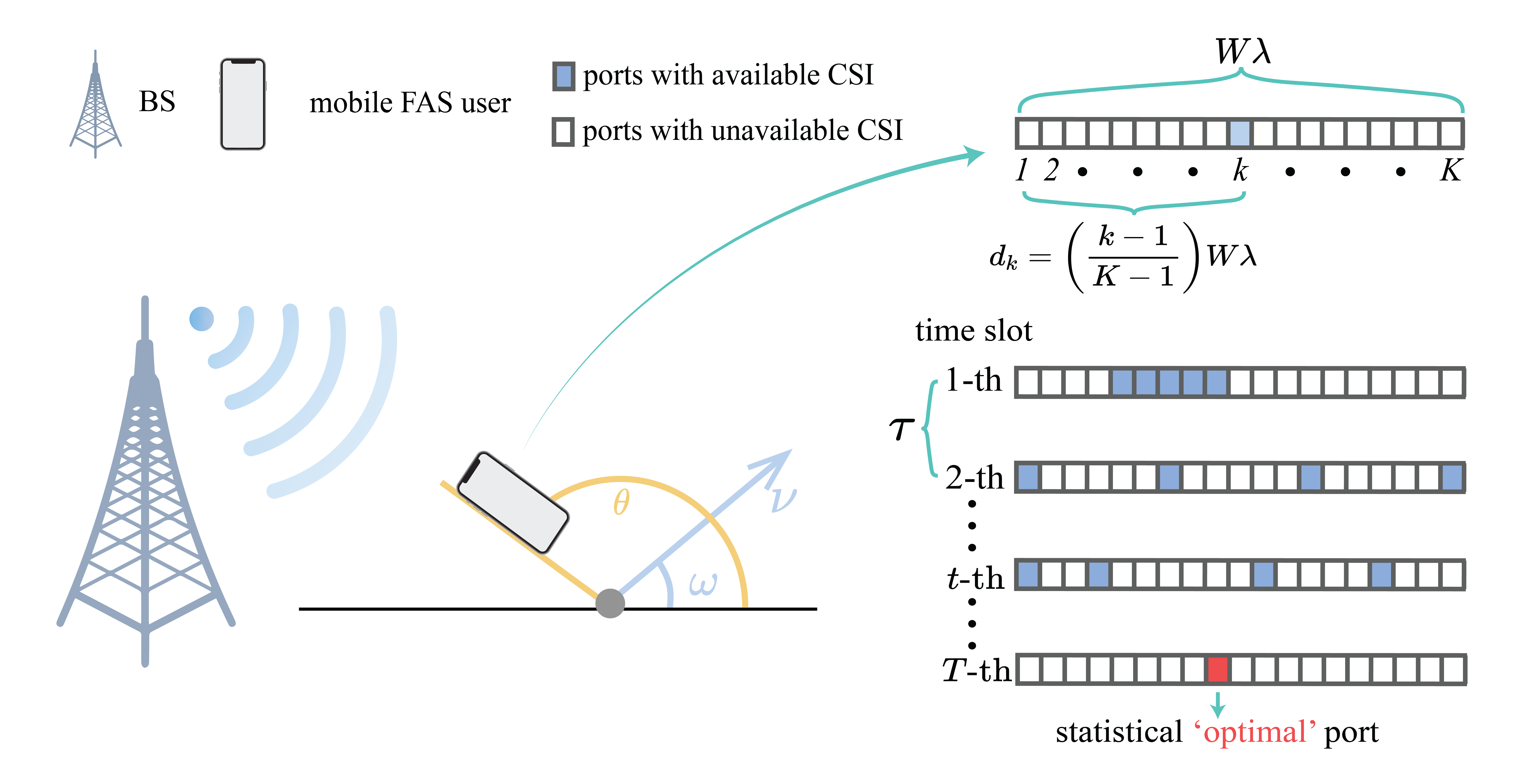}
    \caption{\textcolor{black}{The concept of the proposed semi-blind FAS. The left-hand side illustrates the geometric relationship between the BS and mobile FAS user, while the right-hand side shows the spatial-temporal structure of FAS.}}
    \label{fig:system_model}
\end{figure*}

In this paper, we consider a downlink\footnote{\textcolor{black}{The proposed semi-blind FAS framework is general and can be extended to the uplink; the downlink scenario is considered here for clarity and to facilitate a more coherent presentation.}} system in which the BS is equipped with a single, fixed-position antenna, while the mobile user moves at a velocity $\nu$ and angle $\omega$ with respect to the horizontal axis. The mobile user is equipped with a 1D FAS, oriented at an angle $\theta$ with respect to the horizontal axis. Taking into account the time cost associated with CSI estimation, we assume that CSI is available for only a subset of FAS ports over the previous $(T-1)$ time slots, but not at the $T$-th time slot.  Consequently, according to the incomplete historical port CSI, we evaluate the performance of each FAS port at the $T$-th time slot, thereby finding the statistical \textcolor{black}{optimal} port. Fig.~\ref{fig:system_model} depicts our model.

\subsection{FAS Physical Model}

The 1D FAS consists of $K$ ports that are uniformly distributed along a linear space of length $W\lambda$, where $W$ is the scaling factor and $\lambda$ denotes the communication wavelength. As illustrated on the right-hand side of Fig.~\ref{fig:system_model}, the first port serves as the reference port, relative to which the displacement of $k$-th port is measured, denoted as
\begin{align}
    d_{k} = \left(\frac{k-1}{K-1}\right)W\lambda.
\end{align}
Consequently, the distance between the $k$-th port and $k'$-th port is given by
\begin{align}
    \delta_{k,k'} &= \frac{W\lambda}{K-1}\left|k-k'\right|.
\end{align}
Furthermore, we model the FAS channel in the time domain as a discrete-time stochastic process, evolving over distinct time slots. Accordingly, the time interval between the $t$-th and $t'$-th time slots is given by
\begin{align}
\Delta_{t,t'} = \tau|t - t'|,
\end{align}
where $\tau$ represents the duration of each time slot.\footnote{The duration of the time slot, $\tau$, in this paper, can also be interpreted as the time interval of port switching. For clarity of representation and to provide more insights into the system design, this explanation is omitted.}

\subsection{FAS Channel Model}
Let  $h_{k}^{t}$ denotes the channel of the $k$-th FAS port at the $t$-th time slot, which is modeled as a circularly symmetric complex Gaussian random variable, given by
\begin{align}
    h_{k}^{t} = x_{k}^{t}+jy_{k}^{t},
\end{align}
where both $x_{k}^{t}$ and $y_{k}^{t}$ are zero-mean Gaussian random variables with variance of $\frac{\sigma_{0}^2}{2}$. Consequently, the received signal of the $k$-th FAS port at the $t$-th time slot, is given by
\begin{align}
    z_{k}^{t} = s^{t}h_{k}^{t}+\eta_{k}^{t},
\end{align}
where $s^{t}$ is the transmitted signal at the $t$-th time slot, and $\eta_{k}^{t}$ is the complex additive white Gaussian noise (AWGN) of the $k$-th port at the $t$-th time slot, with zero mean and variance of $\sigma_{\eta}^2$. We consider a constant modulus modulation scheme\footnote{\textcolor{black}{The proposed semi-blind FAS can also be applied to other variable amplitude modulation schemes (e.g., QAM) by disregarding the symbol power in the port selection process. This assumption is adopted here only for the sake of readability throughout the paper.}}, where the power of the information symbol is assumed to remain identical across all time slots, i.e., $\mathrm{E}[|s^{t}|^2] = \sigma_{s}^2$, $\forall t$. Subsequently, the average signal-to-noise ratio (SNR) is defined as
\begin{align}\label{eq:average_SNR}
    \Gamma = \sigma_{0}^2\frac{\sigma_{s}^2}{\sigma_{\eta}^2}.
\end{align}

Let the full set of port indices, sorted in increasing order, be denoted by the set $\mathcal{K} = \{1, \dots, K\}$. The indices of FAS ports for which the CSI is available at the $t$-th time slot are denoted by the set $\bar{\mathcal{K}}^{t} \subseteq \mathcal{K}$, also sorted in increasing order. Similarly, the indices of ports for which the CSI is unavailable at the $t$-th time slot are given by the set $\hat{\mathcal{K}}^{t} = \mathcal{K} \setminus \bar{\mathcal{K}}^{t}$.

Subsequently, we define the row vectors $\boldsymbol{\mathrm{h}}_{\mathcal{\bar{K}}^{t}}^{t}$ and $\boldsymbol{\mathrm{h}}_{\mathcal{\hat{K}}^{t}}^{t}$
to represent the FAS channels with available and unavailable CSI at the $t$-th time slot, respectively, as follows
\begin{align}\boldsymbol{\mathrm{h}}_{\mathcal{\bar{K}}^{t}}^{t}&= \boldsymbol{\mathrm{x}}_{\mathcal{\bar{K}}^{t}}^{t}+j\boldsymbol{\mathrm{y}}_{\mathcal{\bar{K}}^{t}}^{t}, \\
\boldsymbol{\mathrm{h}}_{\mathcal{\hat{K}}^{t}}^{t}&= \boldsymbol{\mathrm{x}}_{\mathcal{\hat{K}}^{t}}^{t}+j\boldsymbol{\mathrm{y}}_{\mathcal{\hat{K}}^{t}}^{t},
\end{align}
where $\boldsymbol{x}_{\bar{\mathcal{K}}^{t}}^{t} = \left[ x_{k}^{t} \right]_{k \in \bar{\mathcal{K}}^t}$, $\boldsymbol{y}_{\bar{\mathcal{K}}^{t}}^{t} = \left[ y_{k}^{t} \right]_{k \in \bar{\mathcal{K}}^t}$, $\boldsymbol{\mathrm{x}}_{\mathcal{\hat{K}}^{t}}^{t} = \left[ x_{k}^{t} \right]_{k \in \hat{\mathcal{K}}^t}$ and $\boldsymbol{\mathrm{y}}_{\mathcal{\hat{K}}^{t}}^{t} = \left[ y_{k}^{t} \right]_{k \in \hat{\mathcal{K}}^t}$ are the in-phase and quadrature components of the corresponding channels, respectively, in the row vector form.

{\color{black}Let $\mathcal{T} = \{1, \dots, T\}$ denote the set of time indices. Define $\bar{\mathcal{T}} \subseteq \mathcal{T}$ as the subset of time indices for which CSI is available for a non-empty subset of ports, and let $\hat{\mathcal{T}} = \mathcal{T} \setminus \bar{\mathcal{T}}$ denote the complementary subset of time indices for which CSI is unavailable for all ports. The number of historical time slots is defined as $\iota \triangleq |\bar{\mathcal{T}}|$. Furthermore, let $|\bar{\mathcal{K}}^{t}|$ denote the number of sampling ports at time slot $t$, corresponding to the number of ports with observed CSI at time slot $t$. Consequently, the FAS channel vectors corresponding to available and unavailable CSI, denoted by $\bar{\boldsymbol{h}}$ and $\hat{\boldsymbol{h}}$, respectively, can be expressed as:
\begin{align}
\bar{\boldsymbol{\mathrm{h}}} & = \boldsymbol{\bar{x}} + j\boldsymbol{\bar{y}},\\    \hat{\boldsymbol{\mathrm{h}}} & = \boldsymbol{\hat{x}}+j\boldsymbol{\hat{y}},
\end{align}
where $\boldsymbol{\bar{x}} = \left[ \boldsymbol{x}_{\bar{\mathcal{K}}^{t}}^{t} \right]_{t \in \bar{\mathcal{T}}}$, $\boldsymbol{\bar{y}} = \left[ \boldsymbol{y}_{\bar{\mathcal{K}}^{t}}^{t} \right]_{t \in \bar{\mathcal{T}}}$, $\boldsymbol{\hat{x}} = \left[ \hat{\boldsymbol{x}} \right]_{t \in \hat{\mathcal{T}}}$, and $\boldsymbol{\hat{y}} = \left[ \hat{\boldsymbol{y}} \right]_{t \in \hat{\mathcal{T}}}$  are the in-phase and quadrature components of the corresponding channels, respectively, in the row vector form.

The central objective of the proposed semi-blind FAS is to evaluate the statistical performance of the FAS channels with unavailable CSI, $\hat{\boldsymbol{h}}$, given the historical CSI of the FAS channels, $\bar{\boldsymbol{h}}$, and therefore determine the optimal port for signal reception. Without loss of generalization\footnote{The proposed framework can be readily extended to other related scenarios. For instance, at a given time slot, one may evaluate the performance of FAS ports with unavailable CSI conditioned on partial CSI availability across a subset of ports. This case is omitted due to space limitations.}, in this paper, we only interested at the $T$-th time slot, that is, we only estimate the statistical performance of each FAS port at the $T$-th time slot. For notational simplicity, we assume that the $T$-th time slot is the only time slot at which CSI is unavailable for all FAS ports. Under this assumption, and based on the notation defined above, we have $\hat{\mathcal{T}} = \{T\}$, $\hat{\mathcal{K}}^{T} = \mathcal{K}$, and $\hat{\boldsymbol{h}} = \left[\boldsymbol{h}_{\hat{\mathcal{K}}^{T}}^{T}\right]$.

Combining the channels with available and unavailable CSI, the FAS spatial-temporal channel can be expressed as
\begin{align}
    \boldsymbol{\mathrm{h}} = \left[\hat{\boldsymbol{\mathrm{h}}},\bar{\boldsymbol{\mathrm{h}}}\right].
\end{align}
Let $\mathcal{L}(\cdot)$ denotes the distribution of random variables. Consequently, $\boldsymbol{\mathrm{h}}$ follows a complex Gaussian distribution as $\mathcal{L}\left(\boldsymbol{\mathrm{h}}\right) = \mathcal{CN}\left(\boldsymbol{0},\sigma_{0}^2\boldsymbol{\Sigma}\right)$, where $\boldsymbol{\Sigma}$ is the FAS spatial-temporal correlation matrix, which will be specified in the following subsection.}

\subsection{FAS Spatial-Temporal Correlation Structure}
The main feature of FAS is the spatial correlated channel, which has been well understood in previous works. However, it is worth noting that the FAS is not only spatial correlated due to its compactness, but also temporal correlated naturally on the time domain. As a result, we have the spatial-temporal correlation coefficient\footnote{\textcolor{black}{The proposed semi-blind FAS is not restricted to the correlation function presented in this paper. It is applicable more broadly, including in non-rich scattering scenarios such as millimeter-wave communications, where the correlation coefficient may not have a closed-form expression but can be empirically measured and subsequently computed.}} between the $k$-th port at the $t$-th time slot, and $k'$-th FAS port at the $t'$-th time slot, expressed as \cite[Eq.(2.129)]{stuber2001principles}
\begin{align}\label{eq:space_time_correlation_function}
    \phi\left(\delta_{k,k'},\Delta_{t,t'}\right) = J_{0}\left(\sqrt{\alpha^2+\beta^2-2\alpha\beta\cos{\left(\theta-\omega\right)}}\right),
\end{align}
where $J_0(\cdot)$ is the zero-order Bessel function of the first kind, $\alpha$ and $\beta$ denote the spatial and temporal correlation parameters, respectively, given by
\begin{equation}\label{eq:alpha_beta}
\left\{\begin{aligned}
\alpha&= 2\pi\frac{ \delta_{k,k'}}{\lambda}, \\
\beta &= 2\pi\frac{\nu\Delta_{t,t'}}{\lambda}.
\end{aligned}\right.
\end{equation}
Consequently, the FAS spatial-temporal correlation matrix, can be expressed as
\begin{align}
    \boldsymbol{\Sigma} = 
    \left[
    \begin{array}{ccccc}
         \boldsymbol{\Sigma}_{\mathcal{\hat{K}}^{T}} &\boldsymbol{\Sigma}_{\mathcal{\hat{K}}^{T},\mathcal{\bar{K}}^{T\text{-}1}} &\cdots &\boldsymbol{\Sigma}_{\mathcal{\hat{K}}^{T},\mathcal{\bar{K}}^{1}}\\ \boldsymbol{\Sigma}_{\mathcal{\hat{K}}^{T},\mathcal{\bar{K}}^{T\text{-}1}}^{\top} & \boldsymbol{\Sigma}_{{\mathcal{\bar{K}}^{T\text{-}1}}} & & \boldsymbol{\Sigma}_{\mathcal{\bar{K}}^{T\text{-}1},\mathcal{\bar{K}}^{1}} \\
         \vdots &  &\ddots & \\ \boldsymbol{\Sigma}_{\mathcal{\hat{K}}^{T},\mathcal{\bar{K}}^{1}}^{\top} & \boldsymbol{\Sigma}_{\mathcal{\bar{K}}^{T\text{-}1},\mathcal{\bar{K}}^{1}}^{\top}& \cdots&\boldsymbol{\Sigma}_{{\mathcal{\bar{K}}^{1}}}
    \end{array}
    \right],
\end{align}
where each sub-matrix is given by
\begin{equation}\label{}
\left\{\begin{aligned}
\left(\boldsymbol{\Sigma}_{\mathcal{\hat{K}}^{T}} \right)_{n,m} &= \phi\left(\delta_{[\mathcal{\hat{K}}^{T}]_{n},[\mathcal{\hat{K}}^{T}]_{m}},0\right),\\
\left(\boldsymbol{\Sigma}_{\mathcal{\bar{K}}^{t}} \right)_{n,m} &= \phi\left(\delta_{[\mathcal{\bar{K}}^{t}]_{n},[\mathcal{\bar{K}}^{t}]_{(m)}},0\right), \\
\left(\boldsymbol{\Sigma}_{\mathcal{\hat{K}}^{T},\mathcal{\bar{K}}^{t}} \right)_{n,m} &= \phi\left(\delta_{[\mathcal{\hat{K}}^{T}]_{n},[\mathcal{\bar{K}}^{t}]_{m}},\Delta_{T,t}\right).
\end{aligned}\right.
\end{equation}
Here, $[\cdot]_{n}$ denotes the $n$-th member of the corresponding set, while $(\cdot)_{n,m}$ represents the element at the $(n,m)$-th position of the corresponding matrix. Sub-matrix $\boldsymbol{\Sigma}_{\mathcal{\hat{K}}^{T}}$ and $\boldsymbol{\Sigma}_{\mathcal{\bar{K}}^{T}}$ include the spatial-correlation of the channels with unavailable and available CSI, respectively. $\boldsymbol{\Sigma}_{\mathcal{\hat{K}}^{T},\mathcal{\bar{K}}^{t}}$ contains the spatial-temporal correlation between the channels with unavailable CSI and channels with available CSI.

It is convenient to represent the spatial-temporal correlation matrix $\boldsymbol{\Sigma}$ in the following composed form as,
\begin{align}
    \boldsymbol{\Sigma} = 
    \left[
    \begin{array}{cc}
         \boldsymbol{\Sigma}_{\hat{\boldsymbol{\mathrm{h}}}} & \boldsymbol{\Sigma}_{\hat{\boldsymbol{\mathrm{h}}},\bar{\boldsymbol{\mathrm{h}}}} \\
         \boldsymbol{\Sigma}_{\hat{\boldsymbol{\mathrm{h}}},\bar{\boldsymbol{\mathrm{h}}}}^{\top} & \boldsymbol{\Sigma}_{\bar{\boldsymbol{\mathrm{h}}}}
    \end{array}
    \right],
\end{align}
where 
\begin{equation}\label{eq:correlation_Matirx}
\left\{\begin{aligned}
\boldsymbol{\Sigma}_{\hat{\boldsymbol{\mathrm{h}}}} &= \boldsymbol{\Sigma}_{\mathcal{\hat{K}}^{T}},\\
\boldsymbol{\Sigma}_{\hat{\boldsymbol{\mathrm{h}}},\bar{\boldsymbol{\mathrm{h}}}} &= \left[ \boldsymbol{\Sigma}_{\mathcal{\hat{K}}^{T},\mathcal{\bar{K}}^{T\text{-}1}},\cdots,\boldsymbol{\Sigma}_{\mathcal{\hat{K}}^{T},\mathcal{\bar{K}}^{1}}\right], \\
\boldsymbol{\Sigma}_{\bar{\boldsymbol{\mathrm{h}}}} &= \left[
    \begin{array}{cccc}
           \boldsymbol{\Sigma}_{{\mathcal{\bar{K}}^{T\text{-}1}}} & \cdots& \boldsymbol{\Sigma}_{\mathcal{\bar{K}}^{T\text{-}1},\mathcal{\bar{K}}^{1}} \\
           \vdots&\ddots &\vdots \\  \boldsymbol{\Sigma}_{\mathcal{\bar{K}}^{T\text{-}1},\mathcal{\bar{K}}^{1}}^{\top}& \cdots&\boldsymbol{\Sigma}_{{\mathcal{\bar{K}}^{t}}}
    \end{array}
    \right].
\end{aligned}\right.
\end{equation}
Here, the correlation among the FAS \textcolor{black}{channels} with unavailable CSI is characterized through $\boldsymbol{\Sigma}_{\hat{\boldsymbol{\mathrm{h}}}}$. Similarly, the correlation among the FAS \textcolor{black}{channels} with available CSI is represented by $\boldsymbol{\Sigma}_{\bar{\boldsymbol{\mathrm{h}}}}$. Finally, the cross correlation between the FAS \textcolor{black}{channels} with available CSI and unavailable CSI is described in $\boldsymbol{\Sigma}_{\hat{\boldsymbol{\mathrm{h}}},\bar{\boldsymbol{\mathrm{h}}}}$.

\section{Semi-Blind FAS}\label{section:partial_FAS}
\textcolor{black}{In this section, we present the achievement of the proposed semi-blind FAS. First, we characterize the conditional distribution of the FAS channel at the target time slot, $\hat{\boldsymbol{\mathrm{h}}}$, given limited historical CSI, $\bar{\boldsymbol{\mathrm{h}}}$, in terms of both the probability density function (PDF) and cumulative distribution function (CDF). Based on these expressions, the outage probability of each FAS port is then derived. Subsequently, the port selection criteria for the semi-blind FAS is developed. We further consider a special case in which only the CSI of a single port is available and derive the corresponding optimal solution.}

\subsection{Conditional Distribution of FAS Channel}
With the aim to provide the conditional distribution of $\hat{\boldsymbol{h}}|\bar{\boldsymbol{h}}$, we found that the conditional distribution on the in-phase component, as well as the quadrature component, is useful in derivation, which is provided in the following lemma.
 
\begin{lemma}\label{lemma:Law_Conditional_Normal}
    Given the observed in-phase and quadrature components, $\bar{\boldsymbol{x}}= \boldsymbol{\mathrm{a}_{x}}$ and $\bar{\boldsymbol{y}}= \boldsymbol{\mathrm{a}_{y}}$, the conditional distributions of $\hat{\boldsymbol{x}}|\bar{\boldsymbol{x}} = \boldsymbol{\mathrm{a}_{x}}$ and $\hat{\boldsymbol{y}}|\bar{\boldsymbol{y}} = \boldsymbol{\mathrm{a}_{y}}$, respectively, are given by
    \begin{align}\mathcal{L}\left(\hat{\boldsymbol{x}}|\bar{\boldsymbol{x}} = \boldsymbol{\mathrm{a}_{x}}\right) = \mathcal{N}\left(\boldsymbol{\mathrm{\mu}}_{\hat{\boldsymbol{x}}|\bar{\boldsymbol{x}}=\boldsymbol{\mathrm{a}_{x}}},\frac{\sigma_{0}^2}{2}\boldsymbol{\Sigma}_{\hat{\boldsymbol{x}}|\bar{\boldsymbol{x}}}\right),
    \\\mathcal{L}\left(\hat{\boldsymbol{y}}|\bar{\boldsymbol{y}} = \boldsymbol{\mathrm{a}_{y}}\right) = \mathcal{N}\left(\boldsymbol{\mathrm{\mu}}_{\hat{\boldsymbol{y}}|\bar{\boldsymbol{y}}=\boldsymbol{\mathrm{a}_{y}}},\frac{\sigma_{0}^2}{2}\boldsymbol{\Sigma}_{\hat{\boldsymbol{y}}|\bar{\boldsymbol{y}}}\right),
    \end{align}
    where \begin{equation}\label{}
    \left\{\begin{aligned}
    \boldsymbol{\mathrm{\mu}}_{\hat{\boldsymbol{x}}|\bar{\boldsymbol{x}}=\boldsymbol{\mathrm{a}_{x}}} &=\boldsymbol{\Sigma}_{\hat{\boldsymbol{\mathrm{h}}},\bar{\boldsymbol{\mathrm{h}}}} \boldsymbol{\Sigma}_{\bar{\boldsymbol{\mathrm{h}}}}^{-1} \boldsymbol{\mathrm{a}_{x}},\\
    \boldsymbol{\mathrm{\mu}}_{\hat{\boldsymbol{y}}|\bar{\boldsymbol{y}}=\boldsymbol{\mathrm{a}_{y}}} &=\boldsymbol{\Sigma}_{\hat{\boldsymbol{\mathrm{h}}},\bar{\boldsymbol{\mathrm{h}}}} \boldsymbol{\Sigma}_{\bar{\boldsymbol{\mathrm{h}}}}^{-1} \boldsymbol{\mathrm{a}_{y}},\\
    \boldsymbol{\Sigma}_{\hat{\boldsymbol{x}}|\bar{\boldsymbol{x}}}&= \boldsymbol{\Sigma}_{\hat{\boldsymbol{y}}|\bar{\boldsymbol{y}}} \\ &=\boldsymbol{\Sigma}_{\hat{\boldsymbol{\mathrm{h}}}} -\boldsymbol{\Sigma}_{\hat{\boldsymbol{\mathrm{h}}},\bar{\boldsymbol{\mathrm{h}}}} \boldsymbol{\Sigma}_{\bar{\boldsymbol{\mathrm{h}}}}^{-1} \boldsymbol{\Sigma}_{\hat{\boldsymbol{\mathrm{h}}},\bar{\boldsymbol{\mathrm{h}}}}^{\top}.
    \end{aligned}\right.
\end{equation}
Here, $\boldsymbol{\Sigma}_{\bar{\boldsymbol{\mathrm{h}}}}^{-1}$ denotes the generalized inverse of $\boldsymbol{\Sigma}_{\bar{\boldsymbol{\mathrm{h}}}}$.
\end{lemma}

\begin{proof}
    See Appendix \ref{proof:Lemma_Conditional_Normal}.
\end{proof}
With the conditional distribution of the in-phase and quadrature component provided, the following theorem characterizes the conditional distribution of $\hat{\boldsymbol{h}}|\bar{\boldsymbol{h}}$.
\begin{theorem}\label{theorem:Law_Conditional_FAS_Channel}
    Given the observed CSI $\bar{\boldsymbol{h}} = \boldsymbol{a}_{c} \triangleq \boldsymbol{a}_x + j\boldsymbol{a}_y$, the conditional distribution of $\hat{\boldsymbol{\mathrm{h}}}{\Big |}\bar{\boldsymbol{\mathrm{h}}}=\boldsymbol{a}_{c}$ is given by
    \begin{align}
        \mathcal{L}\left(\hat{\boldsymbol{\mathrm{h}}}{\Big |}\bar{\boldsymbol{\mathrm{h}}}=\boldsymbol{a}_{c}\right)=\mathcal{CN}\left(\boldsymbol{\mathrm{\mu}}_{\hat{\boldsymbol{\mathrm{h}}}{\Big |}\bar{\boldsymbol{\mathrm{h}}}=\boldsymbol{a}_{c}},\sigma_{0}^2\boldsymbol{\Sigma}_{\hat{\boldsymbol{\mathrm{h}}}|\bar{\boldsymbol{\mathrm{h}}}}\right),
    \end{align}
    where \begin{equation}\label{eq:conditional_correlation_matrix_h}
    \left\{\begin{aligned}
    \boldsymbol{\mathrm{\mu}}_{\hat{\boldsymbol{\mathrm{h}}}|\bar{\boldsymbol{\mathrm{h}}}=\boldsymbol{a}_{c}}&= \boldsymbol{\Sigma}_{\hat{\boldsymbol{\mathrm{h}}},\bar{\boldsymbol{\mathrm{h}}}} \boldsymbol{\Sigma}_{\bar{\boldsymbol{\mathrm{h}}}}^{-1}\left( \boldsymbol{\mathrm{a}_{x}}+j\boldsymbol{\mathrm{a}_{y}}\right), \\ \boldsymbol{\Sigma}_{\hat{\boldsymbol{\mathrm{h}}}|\bar{\boldsymbol{\mathrm{h}}}} &= \boldsymbol{\Sigma}_{\hat{\boldsymbol{\mathrm{h}}}} -\boldsymbol{\Sigma}_{\hat{\boldsymbol{\mathrm{h}}},\bar{\boldsymbol{\mathrm{h}}}} \boldsymbol{\Sigma}_{\bar{\boldsymbol{\mathrm{h}}}}^{-1} \boldsymbol{\Sigma}_{\hat{\boldsymbol{\mathrm{h}}},\bar{\boldsymbol{\mathrm{h}}}}^{\top},
    \end{aligned}\right.
\end{equation}
corresponds to the conditional mean matrix and conditional correlation matrix, respectively.
\end{theorem}
\begin{proof}
    This is the result of Lemma \ref{lemma:Law_Conditional_Normal} and the property of the complex Gaussian distribution.
\end{proof}

\textcolor{black}{Established upon the results presented in Lemma \ref{lemma:Law_Conditional_Normal} and Theorem \ref{theorem:Law_Conditional_FAS_Channel}, we proceed to analyze the conditional distribution of the magnitude of the $k$-th port at the $T$-th time slot, conditioned on $\bar{\boldsymbol{h}} = \boldsymbol{a}_{c}$. To facilitate the exposition, we first introduce the following remark, which characterizes the conditional distributions of the in-phase and quadrature components at the $T$-th time slot, respectively.}

\begin{remark}\label{remark:law_Conditional_Normal_Single}
    The \textcolor{black}{conditional distributions} of $x_{k}^{T}$ and $y_{k}^{T}$, conditioned on  $\bar{\boldsymbol{x}} = \boldsymbol{\mathrm{a}_{x}}$ and $\bar{\boldsymbol{y}} = \boldsymbol{\mathrm{a}_{y}}$ respectively, are given by,
    \begin{align}\mathcal{L}\left(x_{k}^{T}|\bar{\boldsymbol{x}} = \boldsymbol{\mathrm{a}_{x}}\right) &= \mathcal{N}\left(\mu_{x,k},\frac{\textcolor{black}{\rho_{k}}}{2}\sigma_{0}^2\right), \\\mathcal{L}\left(y_{k}^{T}|\bar{\boldsymbol{y}} = \boldsymbol{\mathrm{a}_{y}}\right) &= \mathcal{N}\left(\mu_{y,k},\frac{\textcolor{black}{\rho_{k}}}{2}\sigma_{0}^2\right),
\end{align}
where \begin{equation}\label{}
    \left\{\begin{aligned}
    \mu_{x,k}&= \left(\boldsymbol{\mathrm{\mu}}_{\hat{\boldsymbol{x}}|\bar{\boldsymbol{x}}=\boldsymbol{\mathrm{a}_{x}}}\right)_{\left(k\right)},\\
    \mu_{y,k}&= \left(\boldsymbol{\mathrm{\mu}}_{\hat{\boldsymbol{y}}|\bar{\boldsymbol{y}}=\boldsymbol{\mathrm{a}_{y}}}\right)_{\left(k\right)}, \\ \textcolor{black}{\rho_{k}} &= \left(\boldsymbol{\Sigma}_{\hat{\boldsymbol{x}}|\bar{\boldsymbol{x}}}\right)_{\left(k,k\right)}.
    \end{aligned}\right.
\end{equation}
$\boldsymbol{\mathrm{\mu}}_{\hat{\boldsymbol{x}}|\bar{\boldsymbol{x}}=\boldsymbol{\mathrm{a}_{x}}}$, $\boldsymbol{\mathrm{\mu}}_{\hat{\boldsymbol{y}}|\bar{\boldsymbol{y}}=\boldsymbol{\mathrm{a}_{y}}}$ and $\boldsymbol{\Sigma}_{\hat{\boldsymbol{x}}|\bar{\boldsymbol{x}}}$ are given in Lemma \ref{lemma:Law_Conditional_Normal}.
\end{remark}

\textcolor{black}{With the results presented above, the following theorem characterizes the conditional distribution of the magnitude of the $k$-th FAS port at the $T$-th time slot, $|h_{k}^{T}|$, given the historical port CSI.}
\begin{theorem}\label{theorem:PDF_CDF_Rician}
    \textcolor{black}{Conditioned on the historical port CSI $\bar{\boldsymbol{\mathrm{h}}}=\boldsymbol{a}_{c}$}, the magnitude of the $k$-th FAS port at the $T$-th time slot, $|h_{k}^{T}|$, follows a Rician distribution with the corresponding PDF and CDF given by,
    \begin{align}
        f_{|h_{k}^{T}|{\big |}\bar{\boldsymbol{\mathrm{h}}}=\boldsymbol{a}_{c}}(r) &=\frac{2r}{\textcolor{black}{\rho_{k}}\sigma_{0}^2}\exp{\left\{-\frac{r^2+v_{k}^2}{\textcolor{black}{\rho_{k}}\sigma_{0}^2}\right\}}I_{0}\left(\frac{2rv_{k}}{\textcolor{black}{\rho_{k}}\sigma_{0}^2}\right), \label{eq:PDF_Rician}\\
        F_{|h_{k}^{T}|{\big |}\bar{\boldsymbol{\mathrm{h}}}=\boldsymbol{a}_{c}}(r) &=  1-Q_{1}\left(\frac{\sqrt{2}v_{k}}{\textcolor{black}{\rho_{k}}\sigma_{0}},\frac{\sqrt{2}r}{\textcolor{black}{\rho_{k}}\sigma_{0}}\right), \label{eq:CDF_Rician}
    \end{align}
    where $v_{k} = \sqrt{\mu_{x,k}^2+\mu_{y,k}^2}$, $I_{0}\left(\cdot\right)$ is the modified Bessel function of the first kind with order zero, and $Q_{1}\left(\cdot,\cdot\right)$ is the first order Marcum $Q$-function.

\end{theorem}
\begin{proof}
    According to Remark \ref{remark:law_Conditional_Normal_Single}, the conditional distribution of both $x_{k}^{T}$ and $y_{k}^{T}$ are Gaussian with nonzero means and identical variances. Then, following \cite[A.3.2.5]{stuber2001principles}, we have \eqref{eq:PDF_Rician} and \eqref{eq:CDF_Rician}.
\end{proof}

{\color{black}With the conditional analysis provided above, we first characterize the outage event of the $k$-th FAS port at the $T$-th time slot, conditioned on the historical port CSI $\bar{\boldsymbol{\mathrm{h}}}=\boldsymbol{a}_{c}$. The corresponding outage probability is then obtained in the following corollary. Specifically, the outage event of the $k$-th port at the $T$-th time slot is defined as
\begin{align}
    \mathcal{O} &= \left\{ |h_{k}^{T}|^2 \frac{\sigma_{s}^2}{\sigma_{\eta}^2}< \gamma_{\mathrm{th}} 
    {\Bigg |}\bar{\boldsymbol{\mathrm{h}}} =\boldsymbol{a}_{c}\right\},
\end{align}
where $\gamma_{\mathrm{th}}$ is the pre-set threshold.}

\begin{corollary}\label{corollary:conditional_op}
    Conditioned on the historical port CSI $\bar{\boldsymbol{\mathrm{h}}}=\boldsymbol{a}_{c}$, the outage probability of the $k$-th port at the $T$-th time slot is given by
    \begin{align}\label{eq:conditional_op}
        \mathcal{OP}_{|h_{k}^{T}|{\big |}\bar{\boldsymbol{\mathrm{h}}}=\boldsymbol{a}_{c}} = 1-Q_{1}\left(\frac{\sqrt{2}v_{k}}{\textcolor{black}{\rho_{k}}\sigma_{0}},
        \frac{\sqrt{2}}{\textcolor{black}{\rho_{k}}}\sqrt{\frac{\gamma_{\mathrm{th}}}{\Gamma}}\right).
    \end{align}
\end{corollary}
\begin{proof}
    Substituting $\sqrt{\frac{\gamma_{\mathrm{th}}\sigma_{\eta}^2}{\sigma_{s}^2}}$ into the CDF expression \eqref{eq:CDF_Rician},  and using the average SNR definition in \eqref{eq:average_SNR}, we have the outage probability expression that given in \eqref{eq:conditional_op}.
\end{proof}

\subsection{\textcolor{black}{\textcolor{black}{Semi-Blind FAS} Port Selection Criteria}}\label{subsec:semi-blind-FAS}
{\color{black}With the conditional distribution for each FAS port established in the preceding analysis, we next present the key idea of semi-blind FAS. The central objective of semi-blind FAS is to determine the optimal port via statistical analysis of the conditional distribution of each FAS port. Accordingly, the notion of optimality, or equivalently the design objective of semi-blind FAS, can be adapted to different performance requirements.

To begin with, from a statistical perspective, the CDF provides a complete characterization of the distribution of a random variable and, in wireless communications, directly yields the outage probability. Therefore, the first definition of optimality, and arguably the most interpretable one, is to determine the optimal port via the conditional outage probability, as detailed in the following definition.

\begin{definition}\label{def:optimal_OP}
\textcolor{black}{The optimal port of semi-blind FAS can be defined as the port that minimizes the conditional outage probability of the channel magnitude, given the observed historical port CSI, i.e.,}
\begin{align}\label{eq:optinal_port_OP}
    k^{*}_{\mathrm{OP}}=k_{\mathrm{\textcolor{black}{semi}}\text{-}\mathrm{FAS}}^{T} = \arg\min_{k} \mathcal{OP}_{|h_{k}^{T}|{\big |}\bar{\boldsymbol{\mathrm{h}}}=\boldsymbol{a}_{c}},
\end{align}
    where $\mathcal{OP}_{|h_{k}^{T}|{\big |}\bar{\boldsymbol{\mathrm{h}}}=\boldsymbol{a}_{c}}$ is the conditional outage probability of the $k$-th port at the $T$-th time slot, given in the closed-form expression in Corollary \ref{corollary:conditional_op}.
\end{definition}

Consequently, the outage probability of the proposed semi-blind FAS, based on the optimality criterion given in Definition \ref{def:optimal_OP}, can be expressed as
\begin{align}
    \mathcal{OP}_{\mathrm{\textcolor{black}{semi}}\text{-}\mathrm{FAS}} = \min_{k} \mathcal{OP}_{|h_{k}^{T}|{\big |}\bar{\boldsymbol{\mathrm{h}}}=\boldsymbol{a}_{c}}.
\end{align}

The aforementioned discussion focused on identifying the port that minimizes the outage probability, as stated in Definition~\ref{def:optimal_OP}. However, the proposed semi-blind FAS framework is not restricted to this specific criterion. In the following, we propose to exploit moment-based metrics of the conditional distribution of each FAS port, to determine the optimal port for signal reception.

\begin{definition}\label{def:optimal_mean}
\textcolor{black}{The optimal port of semi-blind FAS can be defined as the port that maximizes the conditional expectation of the channel magnitude, given the observed historical port CSI, i.e.,}
\begin{align}\label{eq:optimal_port_CM}
    k_{\mathrm{MEAN}}^{*}=k_{\mathrm{\textcolor{black}{semi}}\text{-}\mathrm{FAS}}^{T} = \arg\max_{k} \mathbb{E}\!\left[ |h_k^T| \,\middle|\, \bar{\boldsymbol{\mathrm{h}}}=\boldsymbol{a}_{c} \right].
\end{align}
\end{definition}

{\color{blue}\begin{remark}\label{remark:conditional_mean}
    Conditioned on the historical port CSI $\bar{\boldsymbol{\mathrm{h}}}=\boldsymbol{a}_{c}$, the expected magnitude of the channel of the $k$-th port at the $T$-th time slot, is given by
    \begin{align}
        \mathbb{E}\!\left[ |h_k^T| \,\middle|\, \bar{\boldsymbol{\mathrm{h}}}=\boldsymbol{a}_{c} \right] = \frac{\sigma_{0}}{2}\sqrt{\pi\rho_{k}}L_{\frac{1}{2}}\left(-\frac{v_{k}^{2}}{\rho_{k}\sigma_{0}^{2}}\right),
    \end{align}
    where $L_{\frac{1}{2}}\left(\cdot\right)$ is the generalized Laguerre function of order $\frac{1}{2}$.
\end{remark}}

\begin{definition}\label{def:optimal_var}
\textcolor{black}{The optimal port of semi-blind FAS can be defined as the port that minimizes the conditional variance of the channel magnitude, given the observed historical port CSI, i.e.,}
\begin{align}\label{eq:optimal_port_var}
    k_{\mathrm{VAR}}^{*}=k_{\mathrm{\textcolor{black}{semi}}\text{-}\mathrm{FAS}}^{T} = \arg\min_{k} \mathbb{V}\!\left[ |h_k^T| \,\middle|\, \bar{\boldsymbol{\mathrm{h}}}=\boldsymbol{a}_{c} \right].
\end{align}
\end{definition}

{\color{blue}
\begin{remark}\label{remark:conditional_var}
    Conditioned on the historical port CSI $\bar{\boldsymbol{\mathrm{h}}}=\boldsymbol{a}_{c}$, the variance of the channel of the $k$-th port at the $T$-th time slot, is given by
    \begin{align}
        \mathbb{V}\!\left[ |h_k^T| \,\middle|\, \bar{\boldsymbol{\mathrm{h}}}=\boldsymbol{a}_{c} \right] = \rho_{k}\sigma_{0}^{2}+v_{k}^{2}-\frac{\pi}{4}\rho_{k}\sigma_{0}^{2}L_{\frac{1}{2}}^{2}\left(-\frac{v_{k}^{2}}{\rho_{k}\sigma_{0}^{2}}\right).
    \end{align}
\end{remark}}

\begin{definition}\label{def:optimal_MS}
\textcolor{black}{The optimal port of semi-blind FAS can be defined as the port that maximizes the conditional mean-to-standard-deviation ratio of the channel magnitude, given the observed historical port CSI, i.e.,}
\begin{align}\label{eq:optimal_port_MS}
    k_{\mathrm{MS}}^{*}=k_{\mathrm{\textcolor{black}{semi}}\text{-}\mathrm{FAS}}^{T} 
    = \arg\max_{k} 
    \frac{\mathbb{E}\!\left[ |h_k^T| \,\middle|\, \bar{\boldsymbol{\mathrm{h}}}=\boldsymbol{a}_{c} \right]}
         {\sqrt{\mathbb{V}\!\left[ |h_k^T| \,\middle|\, \bar{\boldsymbol{\mathrm{h}}}=\boldsymbol{a}_{c} \right]}}.
\end{align}
\end{definition}

The aforementioned Definitions \ref{def:optimal_OP}, \ref{def:optimal_mean}, \ref{def:optimal_var}, and \ref{def:optimal_MS} serve as illustrative instances of the port selection criteria in semi-blind FAS, representing basic definitions of optimality via statistical analysis. \textcolor{blue}{More generally, a broader range of criteria can be employed to define optimality, including those incorporating higher-order statistical characteristics or composite moment-based statistics to capture more complex distributional properties. Owing to space limitations, such extensions are not elaborated here.}}

\subsection{Special Case: Single Observation}
The calculation process of the proposed semi-blind FAS can be further simplified for certain special scenarios. Among all special cases, the most simplified and arguably the most practically achievable scenario is when only the CSI of a single port, $\bar{k}$, at the $(T-1)$-th time slot is available. Based on this limited information, our goal is to evaluate the performance of all $K$ ports at the $T$-th time slot and \textcolor{black}{thereby identify the optimal port according to Definition \ref{def:optimal_OP}}, as detailed in the following.
\begin{corollary}\label{corollary:OP_given_single_port} Given the available CSI of the $\bar{k}$-th port at the $(T-1)$-th time slot, denoted as $h_{\bar{k}}^{T-1} = a_{c} \triangleq a_{x}+ja_{y}$, the outage probability of the $k$-th port at the $T$-th time slot is simplified as
    \begin{align}\label{eq:OP_given_single_port}
        &\mathcal{OP}_{|h_{k}^{T}|{\big |}h_{\bar{k}}^{T-1}=a_{c}} = \\
        &1-Q_{1}\left(A\sqrt{\frac{\phi^2\left(\delta_{k,\bar{k}},\tau\right)}{1-\phi^2\left(\delta_{k,\bar{k}},\tau\right)}},\frac{R}{\sqrt{1-\phi^2\left(\delta_{k,\bar{k}},\tau\right)}}\right), \nonumber
    \end{align}
    where $A = \frac{\sqrt{2}}{\sigma_{0}}|a_{c}|$ and $R = \sqrt{2\frac{\gamma_{\mathrm{th}}}{\Gamma}}$.
\end{corollary}

\begin{proof}
    See Appendix \ref{proof:OP_given_single_port}.
\end{proof}

It is worth pointing out, conditioned on the historical CSI of a single port, the outage probability at each FAS port depends solely on the correlation between the given port and the estimated port. Accordingly, the optimal correlation coefficient can be determined, as established in the following theorem. For the sake of notational simplicity, we denote $\phi\left(\delta_{k,\bar{k}},\tau\right)$ as $\phi$ throughout the subsequent analysis.

\begin{theorem}\label{theorem:optimal_correlation}
    \textcolor{black}{Conditioned on the single observation $h_{\bar{k}}^{T-1} = a_{c}$, the optimal correlation coefficient in achieving the lowest outage probability}, is given by
    \begin{align}\label{eq:optimal_correlation}
        \phi^{*} = \max_{\phi \in \{\xi, 
        \tilde{\phi}, 1-\xi\}} Q_{1}\left(A\sqrt{\frac{\phi^2}{1-\phi^2}},\frac{R}{\sqrt{1-\phi^2}}\right),
    \end{align}
    where 
    \begin{equation}\label{}
    \left\{\begin{aligned}
    \tilde{\phi}&=  \pm g^{-1}\left(\frac{A}{R}\right),\\
    g(\phi)&= \phi\frac{I_{0}\left(2AR\frac{\phi}{1-\phi^2}\right)}{I_{1}\left(2AR\frac{\phi}{1-\phi^2}\right)}.
    \end{aligned}\right.
    \end{equation}
    \textcolor{black}{Here, $g^{-1}(\cdot)$ denotes the inverse of the function $g(\phi)$, $I_{1}(\cdot)$ represents the modified Bessel function of the first kind of order one, and $\xi \to 0^+$ is an arbitrarily small positive constant.}

\end{theorem}
\begin{proof}
    See Appendix \ref{proof:optimal_correlation}.
\end{proof}

\textcolor{black}{According to the optimal correlation coefficient provided, we can thus find the optimal port of semi-blind FAS in achieving the lowest outage probability, as detailed in the following remark.} 
\begin{remark}\label{remark:optimal_location}
    As for a continues FAS, conditioned on $h_{\bar{k}}^{T-1} = a_{c}$, at the $T$-th time slot, the distance between the given port $\bar{k}$ and the optimal port $k^{*}$ is given by
    \begin{align}
        \delta_{k^{*},\bar{k}} = \phi^{\mathrm{inv}}_{\delta}\left(\phi^{*},\tau\right),
    \end{align}
    where $\phi_{\delta}^{\mathrm{inv}}(\cdot,\cdot)$ denotes the inverse function of the spatial-temporal correlation function given in \eqref{eq:space_time_correlation_function} with respect to $\delta$. Consequently, the optimal location of FAS port is given by
    \begin{align}\label{eq:optimal_port_location}
        d_{k^{*}} = d_{\bar{k}} \pm \delta_{k^{*},\bar{k}},~\text{with}~0\leq d_{k^{*}} \leq W\lambda,
    \end{align}
    where $d_{\bar{k}}$ is the location of the port with available CSI.
\end{remark}

\subsection{\textcolor{blue}{Advantages of Semi-Blind FAS}}
{\color{blue}
The foregoing discussion outlined the central idea of the proposed semi-blind FAS, namely, to analyze the conditional distribution associated with each port in order to identify the statistically optimal port. In the following, we highlight the simplicity of the proposed semi-blind FAS scheme, along with its comparison with deep learning-based methods.

\begin{table*}[!ht]
\centering
\color{blue}
\setlength{\tabcolsep}{3pt}
\caption{\textcolor{black}{Comparison of Semi-Blind FAS and Deep Learning-Based Approaches}}
\label{tab:Comparison}
\renewcommand{\arraystretch}{1.1}
\begin{tabular}{l@{\hspace{0.2em}}cccccc}
\toprule
\textbf{Method} &\textbf{Training Time Complexity} & \textbf{Inference Time Complexity} & \textbf{Space Complexity} & \textbf{Required Hardware} & \textbf{Synchronization Overhead}\\
\midrule
Semi-Blind FAS& 0 & $\mathcal{O}\left((n+1)^2K\right)$  &$\mathcal{O}\left((n+1)^2K\right)$& DSP/FPGA/CPU & NO\\
Deep-Learning Based &$\mathcal{O}\left(ENP\right)$ & $\mathcal{O}\left(P\right)$ & $\mathcal{O}\left(n+P\right)$
& CPU/GPU/TPU &YES\\
\bottomrule
\end{tabular}
\vspace{0.3em}
\footnotesize{\textit{Note:} $E$ denotes the number of training epochs, 
$N$ denotes the number of training samples, 
$P$ denotes the number of model parameters, \\
$n$ denotes the number of available CSI samples, 
and $K$ denotes the number of FAS ports.}
\end{table*}

\subsubsection{Simplicity}
The proposed semi-blind FAS relies solely on the channel correlation matrix and an arbitrary (potentially limited) amount of historical CSI. The correlation matrix can be obtained through appropriately designed pilot transmissions that enable the formation of sample covariance matrices \cite{sanguinetti2019toward}, while historical CSI is inherently available in practical communication systems. Consequently, no additional model training or parameter tuning is required, resulting in the training time complexity of $0$. Furthermore, the computation required to identify the optimal port in the proposed semi-blind FAS scheme admits a closed-form solution, as established in Corollary \ref{corollary:conditional_op} and Remarks \ref{remark:conditional_mean} and \ref{remark:conditional_var}, which depends on the adopted optimality criterion and can therefore be precomputed and tabulated. As a result, the inference time complexity of the semi-blind FAS scheme is $\mathcal{O}\left((n+1)^2 K\right)$, where $n$ denotes the number of available CSI samples and $K$ represents the number of FAS ports. The space complexity is $\mathcal{O}\left((n+1)^2 K\right)$, as it is dominated by the storage requirements of the correlation matrix. Owing to the lightweight computation, the optimal port can be computed simultaneously and independently at both the transmitter and the receiver. This property makes the proposed semi-blind FAS particularly attractive for low complexity and latency sensitive applications.

\subsubsection{Comparison with Deep Learning-Based Approaches}
Deep learning-based approaches typically rely on a large number of trainable parameters and involve complex inference procedures, in addition to a computationally intensive training process. Specifically, the training time complexity can be expressed as $\mathcal{O}\left(ENP\right)$, where $E$, $N$, and $P$ denote the number of training epochs, training samples, and model parameters, respectively. The corresponding inference time complexity is $\mathcal{O}\left(P\right)$, while the space complexity is also $\mathcal{O}\left(P\right)$. Typically, $E$ is on the order of $10^2$--$10^3$, $N$ ranges from $10^4$ to $10^6$, and $P$ is commonly in the range of $10^5$--$10^7$. As a result, the overall training complexity can easily reach $10^{11}$--$10^{16}$ operations. For large-scale architectures inspired by large language models (LLMs)\footnote{\textcolor{blue}{State-of-the-art LLMs can involve over $10^{10}$ parameters, corresponding to over 40 GB of storage for model weights in 32-bit precision, which is comparable in size to a Blu-ray movie, though requiring substantial memory to load and run.}}, $P$ can scale from $10^8$ to $10^9$ or beyond, leading to significantly increased computational and memory demands. Consequently, deep learning-based FAS schemes generally require deployment on devices equipped with strong computational capabilities, or dedicated accelerators. Moreover, due to model complexity and computational cost, when port selection is performed at either the transmitter or the receiver, additional synchronization may be required to inform the other side of the selected port at each time slot, thereby introducing extra communication overhead. On the other hand, deep learning-based approaches offer a complementary advantage. While semi-blind FAS determines the optimal port based on statistical criteria, it does not yield the instantaneous CSI of that port; therefore, conventional channel estimation remains necessary to enable accurate signal decoding. In contrast, learning-based methods can infer approximate CSI directly from historical observations, which may be sufficient for decoding in certain scenarios. From a system perspective, the choice between semi-blind FAS and deep learning-based approaches depends on the system requirements, including computational resources and the desired level of CSI accuracy. Table \ref{tab:Comparison} summarized the aforementioned comparison.}

\section{Physical Insights}\label{section:Physical_Insights}
In this section, we provide some useful insights into the proposed semi-blind FAS. We first introduce the residual entropy power ratio, as a prior measure on the reliable of the proposed port estimation. A closed-form expression of the residual entropy power ratio is provided for the proposed semi-blind FAS. Moreover, we provide the temporal dependence structure of the considered spatial–temporal FAS channel, which determines the number of historical time slots required for port prediction.

\subsection{Prior Evaluation: Residual Entropy Power Ratio}
\textcolor{black}{It is important to determine the required number of ports with available CSI, and the number of historical time slots, to obtain a sufficiently accurate estimation. In other words, given the available historical CSI, to assess the reliability of the estimation.} Inspired by information theory, we first presents the following definition that quantify the uncertainty of random variables in terms of differential entropy, hereafter referred to as entropy.
\begin{definition}[Mutual Information]
    The mutual information between vector random variables $\mathbf{X} \in \mathbb{R}^n$ and $\mathbf{Y} \in \mathbb{R}^m$ is defined as
    \begin{align}\label{eq:mutual_information}
        \mathcal{I}\left(\mathbf{X};\mathbf{Y}\right) =  \mathcal{H}\left(\mathbf{X}\right)-  \mathcal{H}\left( \mathbf{X}|\mathbf{Y}\right),
    \end{align}
    where $\mathcal{H}\left(\mathbf{X}\right)$ is the entropy of $\mathbf{X}$, and $\mathcal{H}\left(\mathbf{X}{\big|}\mathbf{Y}\right)$ denotes the conditional entropy of $\mathbf{X}$ given $\mathbf{Y}$. 
\end{definition}

\textcolor{black}{In information theory, mutual information quantifies the statistical dependence between two random variables by measuring the reduction in uncertainty about one variable after observing the other.} Consequently, mutual information is a good measure on the reliability of the estimation in semi-blind FAS. However, mutual information is not a normalized measure. Therefore, a fair comparison cannot be made using mutual information when comparing two systems with different sizes, e.g., FAS with varying sizes in our problem. \textcolor{black}{To this end, in the following, we presents the definition on the proposed residual entropy power ratio.}

\begin{definition}[Residual Entropy Power Ratio]\label{Def:Residual_entropy_power_ratio}
    \textcolor{black}{Let $\mathbf{X} \in \mathbb{R}^n$ and $\mathbf{Y} \in \mathbb{R}^m$. The residual entropy power ratio of $\mathbf{X}$ given $\mathbf{Y}$ is defined as}
    \begin{align}
        \mathcal{R}\left( \mathbf{X}|\mathbf{Y}\right) = \frac{\mathcal{E}\left(\mathbf{X}{\big|}\mathbf{Y}\right)}{\mathcal{E}\left(\mathbf{X}\right)},
    \end{align}
    where $\mathcal{E}\left(\mathbf{X}\right)$ is the entropy power of $\mathbf{X}$, and $\mathcal{E}\left(\mathbf{X}{\big|}\mathbf{Y}\right)$ denotes the entropy power of $\mathbf{X}$ given $\mathbf{Y}$. 
\end{definition}

The proposed residual entropy power ratio captures the proportion of the uncertainty in $\mathbf{X}$ that remains after conditioning on $\mathbf{Y}$. Intuitively, it measures how much of information in $\mathbf{X}$ is still `residual' or unexplained by $\mathbf{Y}$. \textcolor{black}{Importantly, the residual entropy power ratio is normalized and can therefore enable fair comparisons across systems of different sizes. The following remark establishes the relationship between the residual entropy power ratio and mutual information.} 
\begin{remark}\label{remark:Residual_MI}
     Given the vector random variable $\mathbf{Y} \in \mathbb{R}^m$, the residual entropy power ratio of $\mathbf{X} \in \mathbb{R}^n$, $\mathcal{R}\left( \mathbf{X}|\mathbf{Y}\right)$, can be expressed as
    \begin{align}
        \mathcal{R}\left( \mathbf{X}|\mathbf{Y}\right) =  \exp{ \left\{ -\frac{2}{n}\mathcal{I}\left(\mathbf{X};\mathbf{Y}\right) \right\} },
    \end{align}
    where $\mathcal{I}\left(\mathbf{X};\mathbf{Y}\right)$ is the mutual information between $\mathbf{X}$ and $\mathbf{Y}$, defined in \eqref{eq:mutual_information}.
\end{remark}

With the results provided in Remark \ref{remark:Residual_MI}, the physical interpretation of the proposed residual entropy power ratio $\mathcal{R}\left( \mathbf{X}|\mathbf{Y}\right)$ is evident. Firstly, 
the value of $\mathcal{R}\left( \mathbf{X}|\mathbf{Y}\right)$ lies within the interval $(0,1]$. The residual entropy power ratio $\mathcal{R}\left( \mathbf{X}|\mathbf{Y}\right)\to 0$, when the mutual information $\mathcal{I}\left(\mathbf{X};\mathbf{Y}\right) \to +\infty$, indicates the information contained in $\mathbf{X}$, or equivalently the uncertainty of $\mathbf{X}$, is completely conveyed by $\mathbf{Y}$. In other words, given $\mathbf{Y}$, the residual entropy of $\mathbf{X}$, and therefore the residual information of $\mathbf{X}$, is almost zero. In contrast, when $\mathcal{I}\left(\mathbf{X};\mathbf{Y}\right)=0$, the residual entropy power ratio $\mathcal{R}\left( \mathbf{X}|\mathbf{Y}\right)=1$, the information of $\mathbf{X}$ remains unaffected by the knowledge of $\mathbf{Y}$. To rephrase, given $\mathbf{X}$, the residual entropy of $\mathbf{X}$ remains unchanged. 

Next, we provide the expression of the residual entropy power ratio of the proposed semi-blind FAS.

\begin{theorem}\label{theorem:residual_entropy_power_ratio}
    The residual entropy power ratio of the semi-blind FAS, that is, the residual entropy power ratio of $\hat{\boldsymbol{\mathrm{h}}}$ given $\bar{\boldsymbol{\mathrm{h}}}$, can be expressed as
    \begin{align}\label{eq:residual_entropy_power_ratio_partial_FAS}
        \mathcal{R}\left(\hat{\boldsymbol{\mathrm{h}}}{\big|}\bar{\boldsymbol{\mathrm{h}}}\right) = {\left(\frac{\det{\left(\boldsymbol{\Sigma}_{\hat{\boldsymbol{\mathrm{h}}}|\bar{\boldsymbol{\mathrm{h}}}}\right)}}{\det{\left(\boldsymbol{\Sigma}_{\hat{\boldsymbol{\mathrm{h}}}} \right)}}\right)}^{\frac{1}{K}},
    \end{align}
    where $\boldsymbol{\Sigma}_{\hat{\boldsymbol{\mathrm{h}}}}$ and $\boldsymbol{\Sigma}_{\hat{\boldsymbol{\mathrm{h}}}|\bar{\boldsymbol{\mathrm{h}}}}$
    are given in \eqref{eq:correlation_Matirx} and \eqref{eq:conditional_correlation_matrix_h}, respectively. $\det{(\cdot)}$ denotes the determinant operator.
\end{theorem}
\begin{proof}
    See Appendix \ref{proof:residual_entropy_power_ratio}.
\end{proof}

The results in Theorem \ref{theorem:residual_entropy_power_ratio} provides the useful insight of the residual entropy power ratio. We provide the following form to simplify its computation.
\begin{corollary} \label{corollary:residual_exp_entropy_simplify}
     Given $\bar{\boldsymbol{\mathrm{h}}}$, the residual entropy power ratio of $\hat{\boldsymbol{\mathrm{h}}}$ can be expressed as
     \begin{align}\label{eq:residual_power_entropy_simplify}
        \mathcal{R}\left(\hat{\boldsymbol{\mathrm{h}}}{\big|}\bar{\boldsymbol{\mathrm{h}}}\right) = {\left(\det{\left(\boldsymbol{\mathrm{I}}_{K}-\boldsymbol{\Sigma}_{\hat{\boldsymbol{\mathrm{h}}}} ^{-1}\boldsymbol{\Sigma}_{\hat{\boldsymbol{\mathrm{h}}},\bar{\boldsymbol{\mathrm{h}}}} \boldsymbol{\Sigma}_{\bar{\boldsymbol{\mathrm{h}}}}^{-1} \boldsymbol{\Sigma}_{\hat{\boldsymbol{\mathrm{h}}},\bar{\boldsymbol{\mathrm{h}}}}^{\top}
        \right)}\right)}^{\frac{1}{K}},
    \end{align}
    where $\boldsymbol{\mathrm{I}}_{K}$ denotes the identity matrix of size $K$.
\end{corollary}
\begin{proof}
    See Appendix \ref{proof:residual_exp_entropy_simplify}.
\end{proof}

\textcolor{black}{We now have a prior evaluation that can be utilized as a general and universally applicable optimization objective, in the design of semi-blind FAS. It is worth note, residual entropy power ratio can serve not only as prior evaluation on the accuracy of the proposed semi-blind FAS, but also as a fundamental descriptor in the context of FAS port estimation.}

\textcolor{black}{Specifically, prior studies on port estimation commonly assume that, when a certain fraction of the total port CSI is available (e.g., $10\%$), the CSI of the remaining FAS ports can be accurately estimated. However, the accuracy of port estimation is determined not only by the number or percentage of observed ports, but also by their spatial and temporal positions. For clarity of exposition, we refer to this spatial and temporal distribution of observed ports as the \textit{port sampling strategy}. As demonstrated in Theorem \ref{theorem:residual_entropy_power_ratio}, where the same number of observed ports, when located differently within the spatial-temporal domain, results in different values of $\bar{\boldsymbol{\mathrm{h}}}$, which in turn leads to distinctly different outcomes of residual entropy power ratio. Therefore, we claim that, describing the estimation performance of deep learning models solely in terms of the percentage of observations is imprecise. A more appropriate and informative approach is to evaluate the prediction accuracy with respect to the residual entropy power ratio.} We provide the definition of the optimal port sampling strategy in the following.

\begin{definition}[Optimal Port Sampling Strategy]
The optimal port sampling strategy refers to the selection of sampling ports across spatial and temporal domains that minimizes the residual entropy power ratio, given a fixed number of sampling ports and sampling time slots.
\end{definition}

\textcolor{black}{The definition of the optimal port sampling strategy is presented. Nevertheless, finding a solution to this problem is outside the scope of the current study. An example of the optimal port sampling strategy, obtained via exhaustive search in a small system, is provided in Section \ref{section:results}.}

\subsection{Temporal Dependence Structures: Markov and Independence Conditions}

From an estimation perspective, access to a larger amount of historical CSI generally improves prediction accuracy. However, a key question is whether retaining the entire CSI history is necessary for accurate port estimation. In particular, it is important to determine the number of historical time slots sufficient for reliable port prediction, or equivalently, the temporal range beyond which historical CSI no longer contributes \textcolor{black}{effectively} to estimation accuracy.

To address this problem, the following theorem is presented, establishing that, under certain conditions, the considered FAS spatial-temporal channel can be approximated as a Markov process. In order to provide a general result regarding the Markov property of the FAS in the time domain, it is assumed that perfect CSI is available on all FAS ports at each historical time slot, while no CSI is available at the desired $T$-th time slot.

\begin{theorem}\label{theorem:FAS_Markov}
    The FAS spatial-temporal channel structure approximates the Markov process, in the sense that the conditional independence holds as
    \begin{align}\label{eq:FAS_Markov_condition}
    \frac{\left\| (\boldsymbol{\Sigma}_{{\mathcal{\hat{K}}^{T}},{\mathcal{\bar{K}}^{T_{0}-1}} \mid\mathcal{\bar{K}}^{T_{0}}})_{1:K,1+K:2K}\right\|_{L^1}}{K^2}\textcolor{black}{\to} 0,
\end{align}
where $T_{0}$ denotes the start index of the historical CSI, $(\boldsymbol{\Sigma}_{{\mathcal{\hat{K}}^{T}},{\mathcal{\bar{K}}^{T_{0}-1}} \mid\mathcal{\bar{K}}^{T_{0}}})_{1:K,1+K:2K}$ denotes the submatrix formed by the rows indexed from $1$ to $K$ and the columns indexed from $1+K$ to $2K$ of the conditional correlation matrix $\boldsymbol{\Sigma}_{{\mathcal{\hat{K}}^{T}},{\mathcal{\bar{K}}^{T_{0}-1}} \mid
 \mathcal{\bar{K}}^{T_{0}}}$, and $\left\| \cdot\right\|_{L^1}$ denotes the $L^1$ norm. The conditional correlation matrix is computed as
\begin{align}\label{eq:conditional_correlation_matrix_Markov}
&\boldsymbol{\Sigma}_{{\mathcal{\hat{K}}^{T}},{\mathcal{\bar{K}}^{T_{0}-1}} \mid
 \mathcal{\bar{K}}^{T_{0}}}= \boldsymbol{\Sigma}_{{\mathcal{\hat{K}}^{T}},{\mathcal{\bar{K}}^{T_{0}-1}}}- \nonumber\\&\quad\boldsymbol{\Sigma}_{({\mathcal{\hat{K}}^{T}},{\mathcal{\bar{K}}^{T_{0}-1}}),\mathcal{\bar{K}}^{T_{0}}}\boldsymbol{\Sigma}_{{\mathcal{\bar{K}}^{T_{0}}}}^{-1}\boldsymbol{\Sigma}_{({\mathcal{\hat{K}}^{T}},{\mathcal{\bar{K}}^{T_{0}-1}}),\mathcal{\bar{K}}^{T_{0}}}^{\top}.
\end{align}
Each component in \eqref{eq:conditional_correlation_matrix_Markov} is given by
\begin{equation}\label{eq:M_matrix_component}
\left\{\begin{aligned}
\boldsymbol{\Sigma}_{{\mathcal{\hat{K}}^{T}},{\mathcal{\bar{K}}^{T_{0}-1}}} &=  \left[
    \begin{array}{cc}
         \boldsymbol{\Sigma}_{{\mathcal{\hat{K}}^{T}}}   &\boldsymbol{\Sigma}_{{\mathcal{\hat{K}}^{T}},{\mathcal{\bar{K}}^{T_{0}-1}}}\\ \boldsymbol{\Sigma}_{{\mathcal{\hat{K}}^{T}},{\mathcal{\bar{K}}^{T_{0}-1}}}^{\top}& \boldsymbol{\Sigma}_{\mathcal{\bar{K}}^{T_{0}-1}}
    \end{array}
    \right],\\
\boldsymbol{\Sigma}_{({\mathcal{\hat{K}}^{T}},{\mathcal{\bar{K}}^{T_{0}-1}}),\mathcal{\bar{K}}^{T_{0}}} &=  \left[ \boldsymbol{\Sigma}_{{\mathcal{\hat{K}}^{T}},{\mathcal{\bar{K}}^{T_{0}}}},\boldsymbol{\Sigma}_{{\mathcal{\bar{K}}^{T_{0}-1}},{\mathcal{\bar{K}}^{T_{0}}}}\right]^{\top}.
\end{aligned}\right.
\end{equation}
\end{theorem}
\begin{proof}
    See Appendix \ref{proof:FAS_Markov}.
\end{proof}

Another straightforward criterion for determining the temporal range that contributes to estimation accuracy is the unconditional independence. This condition implies that the CSI values at different time slots are statistically independent, and therefore, historical CSI provides no useful information for predicting future FAS channel states. The following remark specifies the condition under which the FAS spatial-temporal channel can be approximated as temporally independent from previous time slots.
\begin{remark}
    The FAS spatial-temporal channel at the $T$-th time slot is temporally independent from the $T_{0}$-th time slot, under the condition that
    \begin{align}\label{eq:FAS_Independent_condition}
    \frac{\left\| (\boldsymbol{\Sigma}_{{\mathcal{\hat{K}}^{T}} \mid\mathcal{\bar{K}}^{T_{0}}})_{1:K,1+K:2K}\right\|_{L^1}}{K^2}\textcolor{black}{\to} 0,
\end{align}
where 
\begin{align}
&\boldsymbol{\Sigma}_{{\mathcal{\hat{K}}^{T}} \mid
 \mathcal{\bar{K}}^{T_{0}}}= \boldsymbol{\Sigma}_{{\mathcal{\hat{K}}^{T}}}-\boldsymbol{\Sigma}_{{\mathcal{\hat{K}}^{T}},\mathcal{\bar{K}}^{T_{0}}}\boldsymbol{\Sigma}_{{\mathcal{\bar{K}}^{T_{0}}}}^{-1}\boldsymbol{\Sigma}_{{\mathcal{\hat{K}}^{T}},\mathcal{\bar{K}}^{T_{0}}}^{\top}.
\end{align}
\end{remark}

\section{Numerical Results and Discussion}\label{section:results}

Here, we provide simulation results to evaluate the performance of the proposed semi-blind FAS. In the simulations, we set $\sigma^2 = \sigma_{s}^2 = 1$, $\sigma_{\eta}^2 = 10^{-1}$. The semi-blind FAS related parameters are set as $\omega=0$, $\theta = \frac{\pi}{2}$, $W=0.5$, $\nu = 15~\mathrm{m/s}$ and $\iota = 1$ unless otherwise specified. The network parameters are configured to meet the requirements of sub-6~GHz band in 5G \cite{3gpp38_211}, with a wavelength of $\lambda = 0.1~\mathrm{m}$, and a time slot duration of $\tau = 2.5\times 10^{-4}~\mathrm{s}$, unless noted otherwise. \textcolor{black}{Although the definition of optimality may vary, as discussed in Section \ref{subsec:semi-blind-FAS}, for clarity and interpretability, the numerical results presented here are obtained based on the definition that minimizes the outage probability, as formalized in Definition \ref{def:optimal_OP}.}

\begin{figure}[!htbp]
    \centering
    \includegraphics[width=1.0\linewidth]{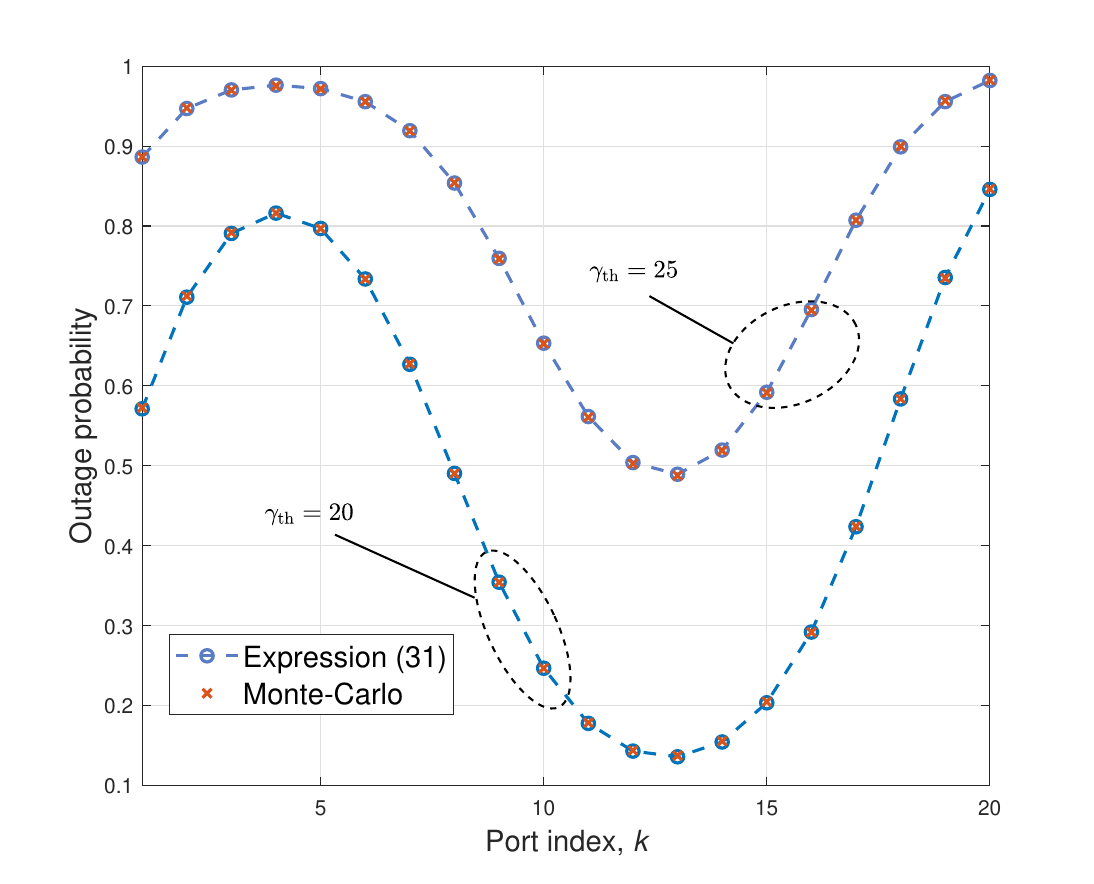}
    \caption{Estimated outage probability of each FAS port at the $T$-th time slot, given different threshold $\gamma_{\mathrm{th}}$ when $K=20$, $\tau = 10^{-3}~\mathrm{s}$, $\bar{\mathcal{K}}^{T\text{-}1} = \{5,10,15\}$ and $|\boldsymbol{a}_{c}| = [1.33,1.56,1.30]$.}
    \label{Fig:OP_Index_Monte}
\end{figure}

\begin{figure}[!htbp]
    \centering
    \includegraphics[width=1.0\linewidth]{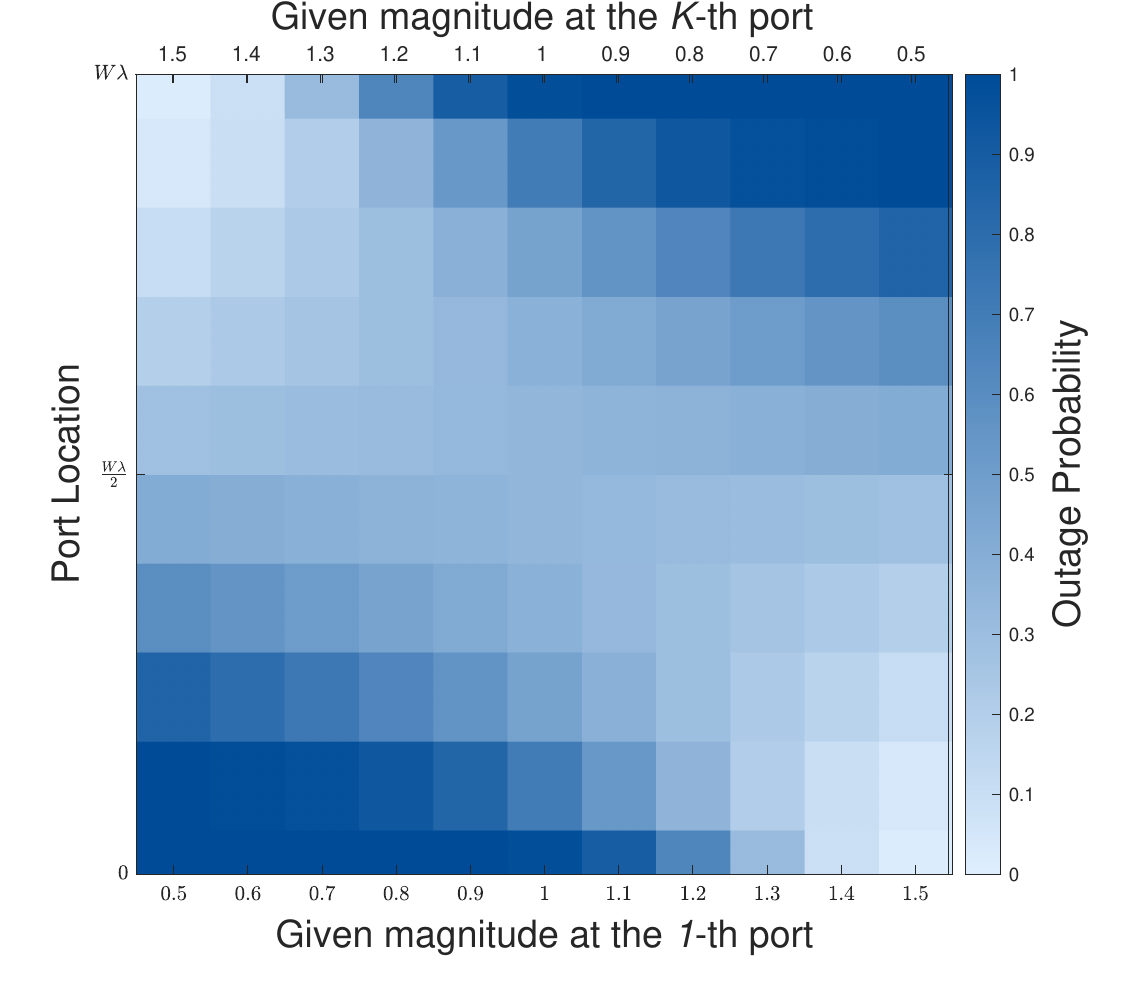}
    \caption{Estimated outage probability of each FAS port against the given magnitude, when $K=10$ and $\gamma_{\mathrm{th}} = 15$.}
    \label{Fig:OP_Given_Magnitude_Port_Location}
\end{figure}

Fig.~\ref{Fig:OP_Index_Monte} provides numerical results for the estimated outage probability of each FAS port at the $T$-th time slot. Firstly, the results confirm that the analytical expressions \eqref{eq:conditional_op} align closely with the Monte-Carlo results, validating
the accuracy of the proposed semi-blind FAS. Secondly, increasing the threshold has a minor effect on the shape of the curve; instead, it primarily results in an increase of the outage probability on each port. Moreover, conditioned on the given historical CSI, there exists a port with the lowest outage probability among all FAS ports, i.e., the statistical optimal port.

\begin{figure}
    \centering
    \includegraphics[width=1.0\linewidth]{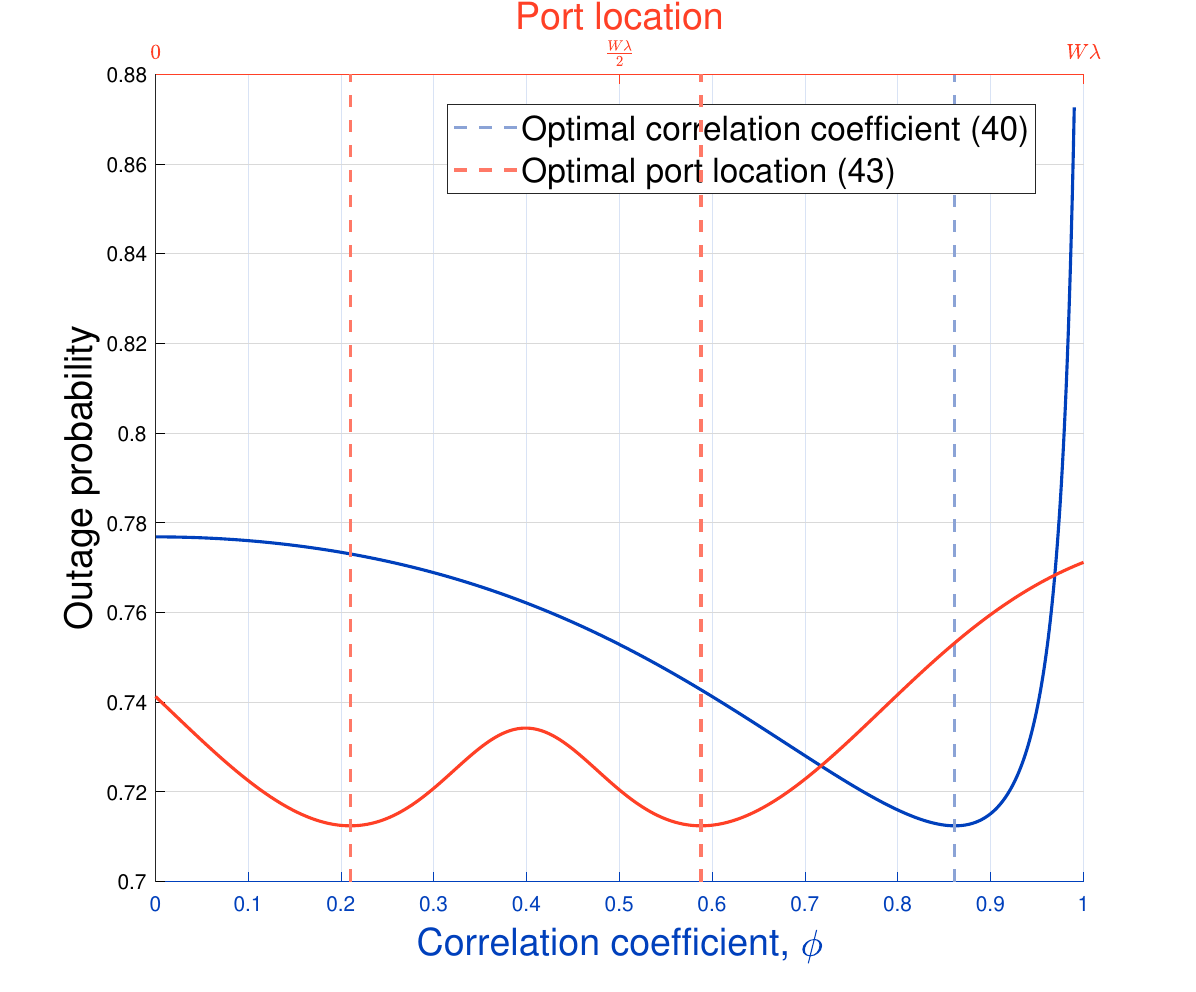}
    \caption{Estimated outage probability against the correlation coefficient and FAS port location, when $K=500$, $\tau =2.5\times 10^{-4}~\mathrm{s}$, $\gamma_{\mathrm{th}}=15$,
    $\bar{k} = 200$ and $|a_{c}| = 1.12$.}
    \label{Fig:Optimzal_Correlation_Coefficient}
\end{figure}

\begin{figure*}[!htbp]
    \centering
  \subfloat[$\tau=10^{-5}~\mathrm{s}$~{\color{black}(high temporal resolution)}\label{fig:spatial_a}]{\includegraphics[width=0.5\linewidth]{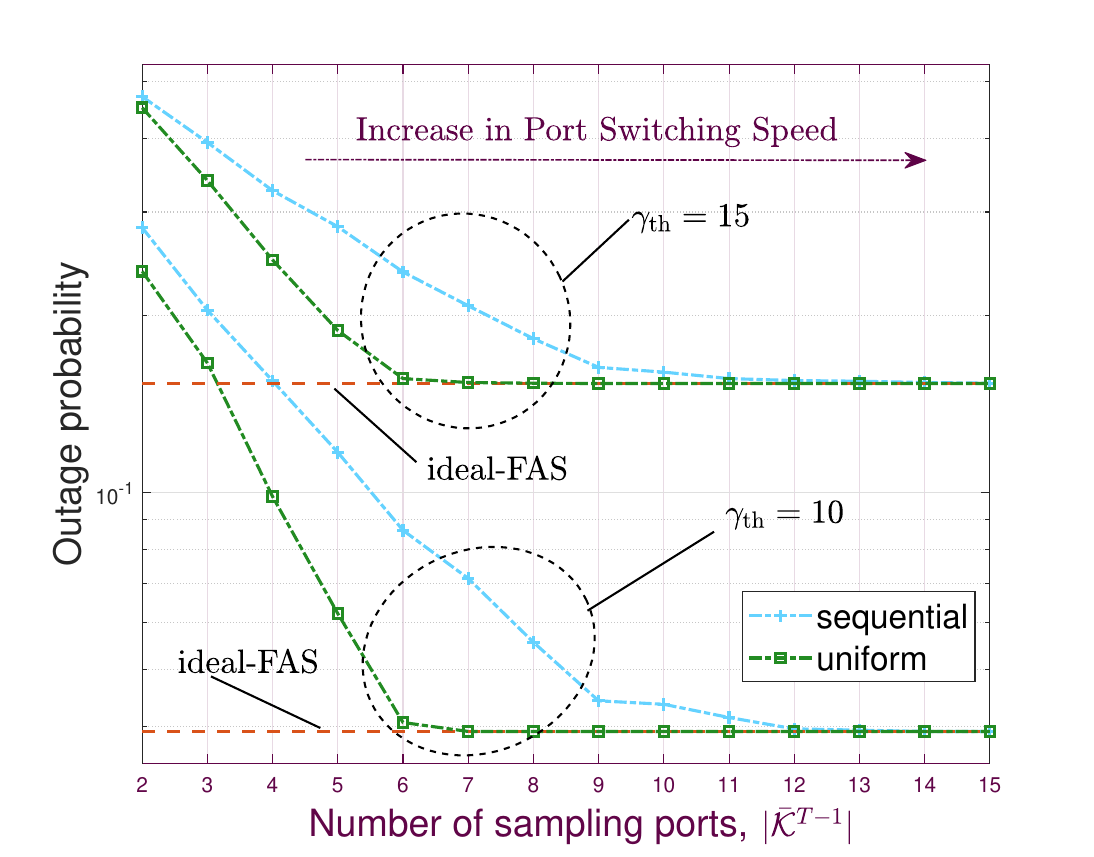}}
    \hfill
    \subfloat[$\tau=2.5\times10^{-4}~\mathrm{s}$~{\color{black}(low temporal resolution)}\label{fig:spatial_b}]{\includegraphics[width=0.5\linewidth]{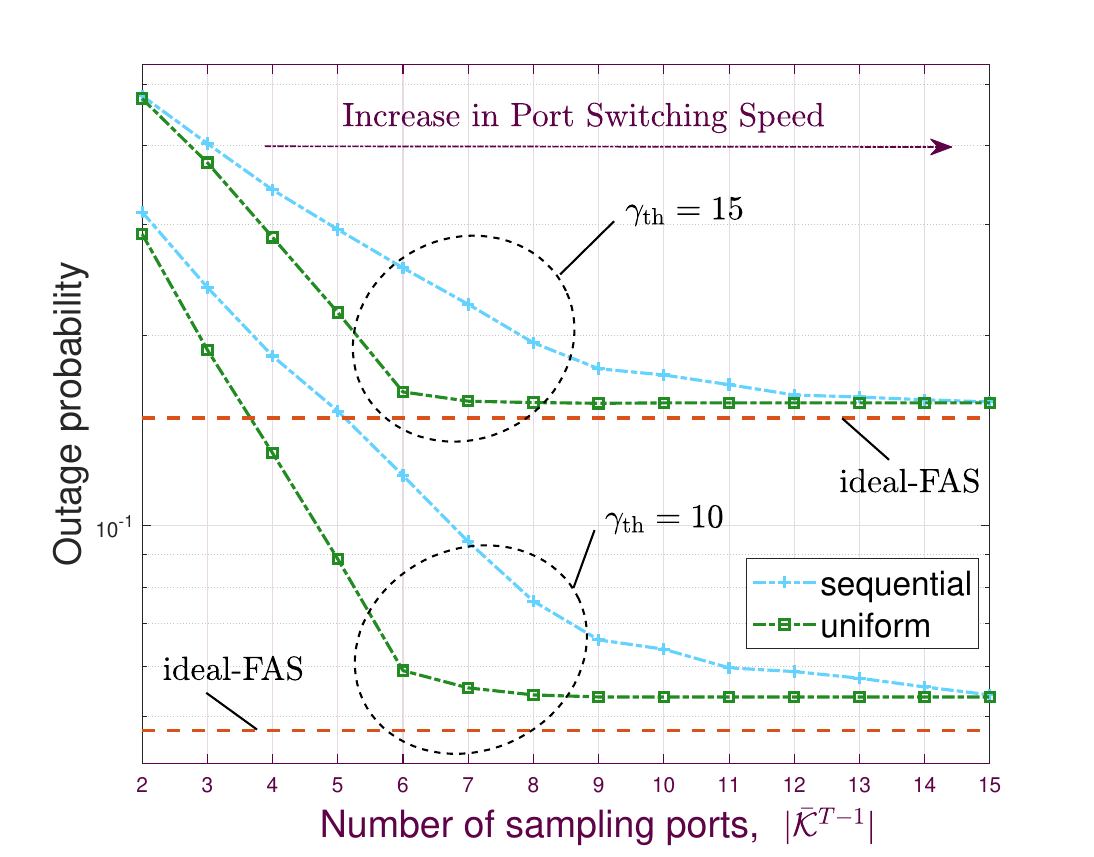}}
  \caption{Outage probability of semi-blind FAS, against the number of sampling ports at the $(T-1)$-th time slot, $|\bar{\mathcal{K}}^{T-1}|$, under different time interval $\tau$, threshold $\gamma_{\mathrm{th}}$ and port sampling strategy, when $W=2$, $K=30$ and $\nu =10 ~\mathrm{m/s}$.}\label{Fig:OP_sampling_ports} 
\end{figure*}

\begin{figure*}[!htb] 
    \centering
  \subfloat[$\tau=10^{-5}~\mathrm{s}$~{\color{black}(high temporal resolution)}\label{fig:temporal_a}]{\includegraphics[width=0.5\linewidth]{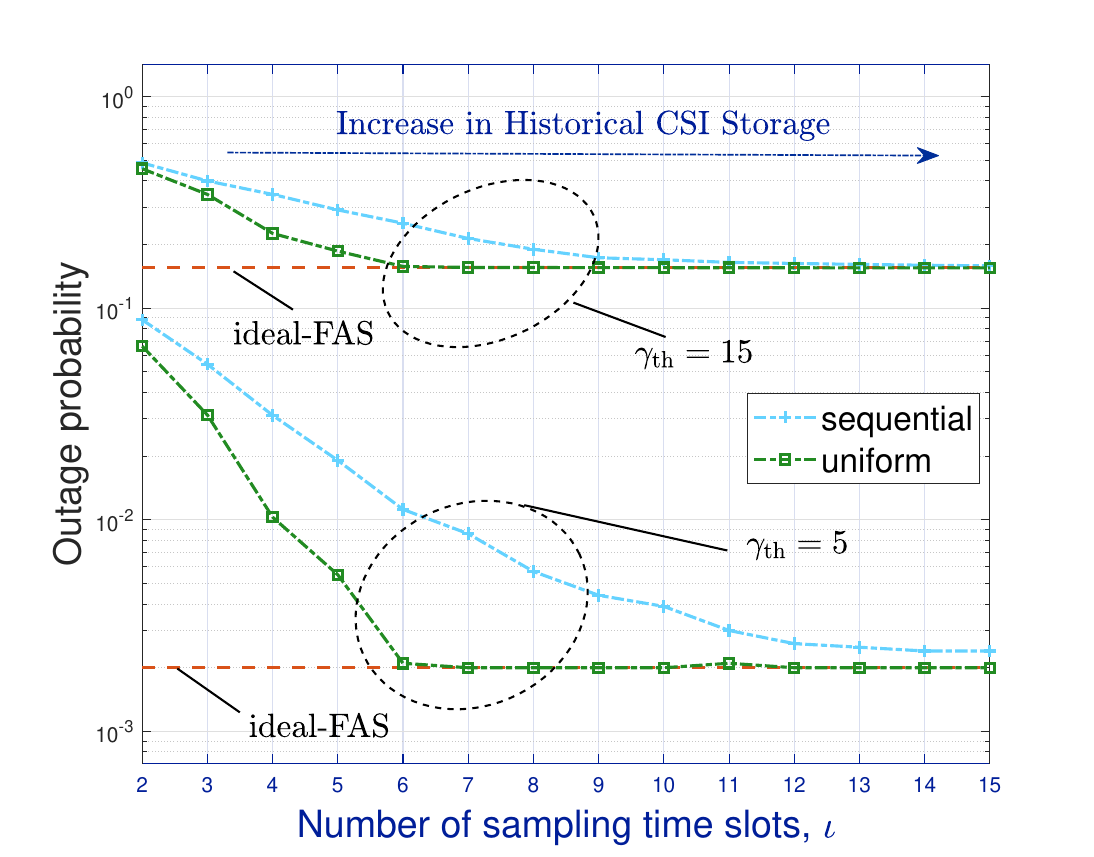}}
    \hfill
    \subfloat[$\tau=10^{-4}~\mathrm{s}$~{\color{black}(low temporal resolution)}\label{fig:temporal_b}]{\includegraphics[width=0.5\linewidth]{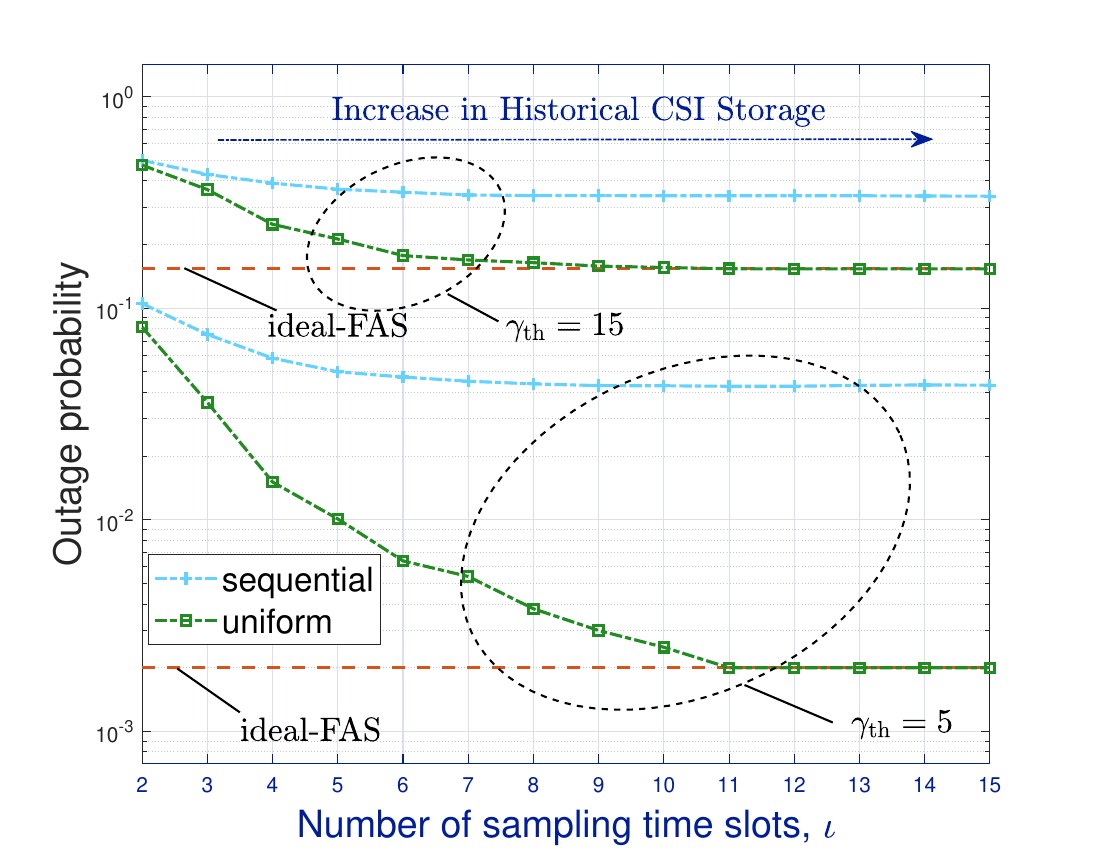}}
  \caption{Outage probability of semi-blind FAS, against the number of sampling time slots $\iota$, under different time interval $\tau$, threshold $\gamma_{\mathrm{th}}$ and port sampling strategy, when $W=2$, $K=30$ and $\nu=10~\mathrm{m/s}$.}\label{Fig:OP_sampling_time_slot} 
\end{figure*}

To analysis further the parameters that influence the location of the statistical optimal port, we refer to the results provided in Fig.~\ref{Fig:OP_Given_Magnitude_Port_Location}. As the x-axis progresses to the right, the region corresponding to lower outage probability gradually shifts from being closer to the $K$-th port towards the $1$-th port. Meanwhile, in the central region of the x-axis, the optimal port tends to be located near the midpoint of the FAS. This phenomenon results from the variation in the given magnitude. Specifically, a port exhibiting a larger magnitude at the $(T-1)$-th time slot is more likely to attain a higher magnitude at the $T$-th time slot. This is due to the naturally high correlation between magnitudes of the same spatial location of FAS over the time domain. Conversely, when the observed magnitudes at both the $1$-th and $K$-th ports at the $(T-1)$-th time slot are low, it implies that the vicinity of these ports is more likely to experience severe fading at the $T$-th time slot. Consequently, ports that are simultaneously distant from both the $K$-th and $1$-th ports achieve a lower outage probability at the $T$-th time slot.

To provide more insights into the statistical optimal location of the proposed semi-blind FAS in achieving the lowest outage probability, we now refer to Fig.~\ref{Fig:Optimzal_Correlation_Coefficient}, which illustrate the optimal correlation coefficient and the corresponding optimal port location based on a single observation. Using expressions \eqref{eq:optimal_correlation} and \eqref{eq:optimal_port_location}, here, only the positive correlation coefficient is reported for clarity of presentation. It is observed that expression \eqref{eq:optimal_correlation} accurately determines the optimal correlation coefficient, while expression \eqref{eq:optimal_port_location} identifies the corresponding optimal port location accordingly. Furthermore, it is found that, in a statistical sense, two ports emerge as the statistical optimal ports. This occurs because, these two ports, exhibit identical correlation coefficients relative to the single observed port, resulting in equivalent statistical performance. 

With the results provided in Fig.~\ref{Fig:OP_Index_Monte},~\ref{Fig:OP_Given_Magnitude_Port_Location} and~\ref{Fig:Optimzal_Correlation_Coefficient} evaluate the outage probability of each semi-blind FAS port, we now study the performance on of the semi-blind FAS on the system level. The results in Fig.~\ref{Fig:OP_sampling_ports} provides the outage probability of the proposed semi-blind FAS against the number of sampling ports, under different port sampling strategy and physical parameters, generated using Monte-Carlo. \textcolor{black}{The ideal-FAS uses full port CSI to determine the optimal port for signal reception, while semi-blind FAS evaluates the statistical optimal port based on the incomplete historical port CSI. In sequential sampling, ports are selected in ascending order starting from the first, while uniform sampling ensures an even distribution across all available ports.} As illustrated in Fig.~\ref{Fig:OP_sampling_ports}\subref{fig:spatial_a}, in the context of \textcolor{black}{high temporal resolution} systems, characterized by short time interval ($\tau=10^{-5}~\mathrm{s}$), with the increase on the number of sampling ports,  the outage probability of semi-blind FAS decreased until it matched that of ideal-FAS for both $\gamma_{\mathrm{th}}=10$ and $\gamma_{\mathrm{th}}=15$, for both sequential and uniform sampling. As for larger time interval ($\tau=2.5\times10^{-4}~\mathrm{s}$), illustrated in Fig.~\ref{Fig:OP_sampling_ports}\subref{fig:spatial_b}, semi-blind FAS can only achieve comparable performance to ideal-FAS, under different threshold $\gamma_{\mathrm{th}}$ and port sampling strategy. This effect arises from the fact that a larger value of $\tau$ leads to reduced time domain correlation, which in turn degrades the reliability of the estimation. Notably, uniform sampling requires fewer sampling ports to achieve convergence in outage probability compared to sequential sampling, under different threshold and time interval.

Despite the results presented in Fig.~\ref{Fig:OP_sampling_ports} demonstrate that the semi-blind FAS can achieve performance comparable to, or even identical with, that of ideal-FAS, as the number of sampling ports increases, this improvement comes with a trade-off. Specifically, increasing the number of sampling ports necessitates a corresponding increase in the required port switching speed, which can be challenging to achieve when the time interval $\tau$ is small. We aims to find a solution that semi-blind FAS can approach the performance of ideal-FAS, \textcolor{black}{but without imposing stringent hardware requirements on the FAS implementation.}

To that end, in Fig.~\ref{Fig:OP_sampling_time_slot}, we consider relaxing the port switching speed requirement, \textcolor{black}{such that the FAS operates at a sufficiently low switching speed and the hardware can support CSI estimation for only a single port at each time slot.} The results in Fig.~\ref{Fig:OP_sampling_time_slot} illustrate how the outage probability of semi-blind FAS varies as the number of sampling time slots, changes. In the \textcolor{black}{high temporal resolution} system ($\tau = 10^{-5}~\mathrm{s}$), as illustrated in Fig.~\ref{Fig:OP_sampling_time_slot}\subref{fig:temporal_a}, with uniform sampling, semi-blind FAS achieve identical outage probability to ideal-FAS, under both $\gamma_{\mathrm{th}} = 5$ and $\gamma_{\mathrm{th}} = 15$. In contrast, sequential sampling semi-blind FAS can only achieve the same outage probability to ideal-FAS when $\gamma_{\mathrm{th}} = 15$, and comparable performance when $\gamma_{\mathrm{th}} = 5$. Moreover, the results in Fig.~\ref{Fig:OP_sampling_time_slot}\subref{fig:temporal_b} provide that, in the \textcolor{black}{low temporal resolution} system ($\tau=10^{-4}~\mathrm{s}$), with the increase on the number of sampling time slots, only uniform sampling can achieve the outage probability identical to that of ideal-FAS, whereas the performance of sequential sampling remains significantly inferior to ideal-FAS, under the considered threshold. This highlights the importance of the port sampling strategy when multiple historical time slots are considered.

To sum up, in the \textcolor{black}{high temporal resolution} system, with the increase on the historical port CSI, semi-blind FAS can approach the performance of ideal-FAS, regardless of the SNR requirement and port sampling strategy. However, in the \textcolor{black}{low temporal resolution} system, the suitable port sampling strategy is essential for semi-blind FAS to effectively approach the performance of ideal-FAS. \textcolor{black}{It should be noted that, even in scenarios where semi-blind FAS cannot achieve the performance of ideal FAS, the semi-blind FAS remains operational and can still benefit from spatial diversity gain, i.e., FAS operates under incomplete CSI.}

It is worth pointing out, the results provided in Fig.~\ref{Fig:OP_sampling_time_slot}, does not rely on the increase of the port switching speed. In fact, with the increase on the number of sampling time slots, the required port switching speed remains the same. The only extra cost incurred is the storage of several historical CSI. It can be understood as, \textcolor{black}{with the} spatial-temporal FAS framework, the semi-blind FAS offloads the CSI estimation burden among different time slots. This benefit is significant important when the number of FAS ports large, which is the key of FAS performance improvement \cite{Fluid_antenna_system,fast_FAMA,slow_FAMA}. These results indicate that, the proposed semi-blind FAS can achieve comparable, even identical performance to that of ideal-FAS, while significantly reduces the required port switching speed, with a negligible cost of historical CSI storage.

\begin{figure}[!ht] 
    \centering
  \subfloat[Residual entropy power ratio\label{fig:R_samping_ports}]{\includegraphics[width=0.50\linewidth]{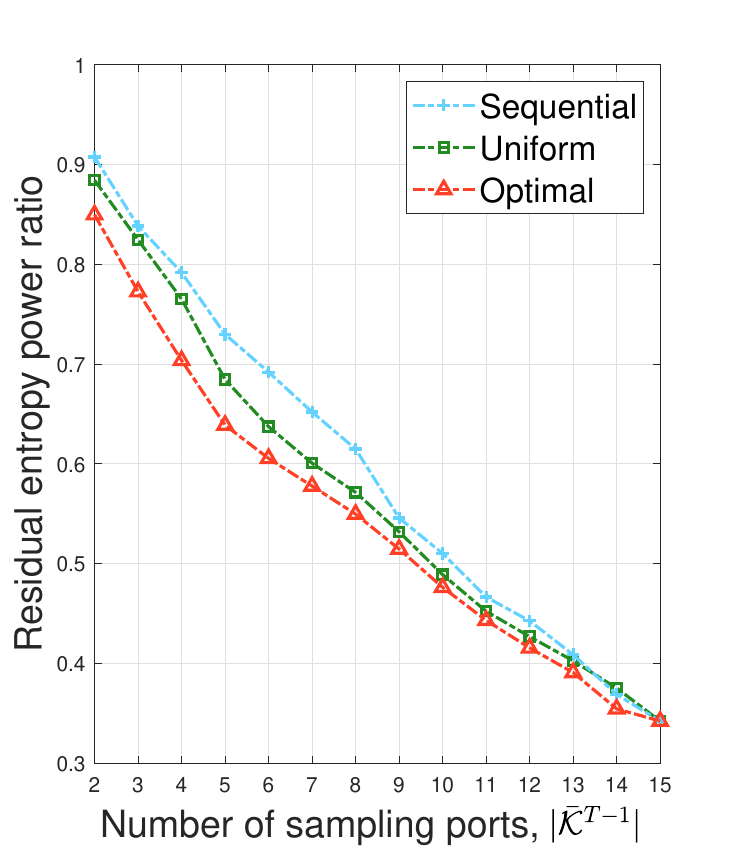}}
    \hfill
  \subfloat[Outage probability\label{fig:OP_samping_ports}]{\includegraphics[width=0.50\linewidth]{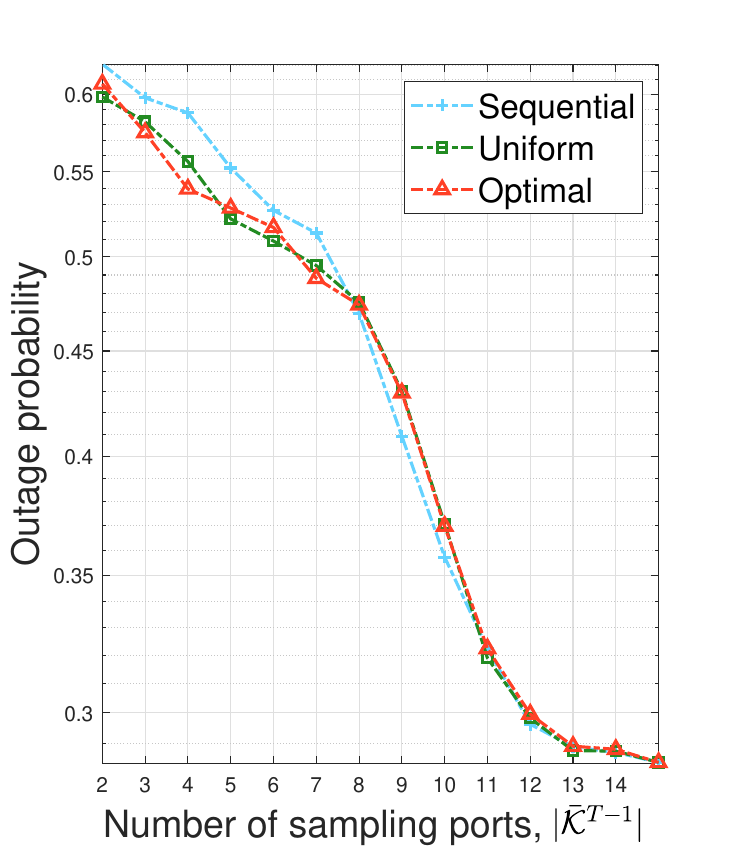}}
  \caption{Different performance metrics against the number of sampling ports at the $(T-1)$-th time slot, $|\bar{\mathcal{K}}^{T-1}|$, given different port sampling strategy, when $W=2$, $K=15$, $\gamma_{\mathrm{th}}=10$, $\tau = 5\times10^{-3}~\mathrm{s}$ and $\nu = 30~\mathrm{m/s}$.}\label{Fig:Residual_Entropy_Activited_port}
\end{figure}

\begin{figure}[]
    \centering
    \includegraphics[width=1.0\linewidth]{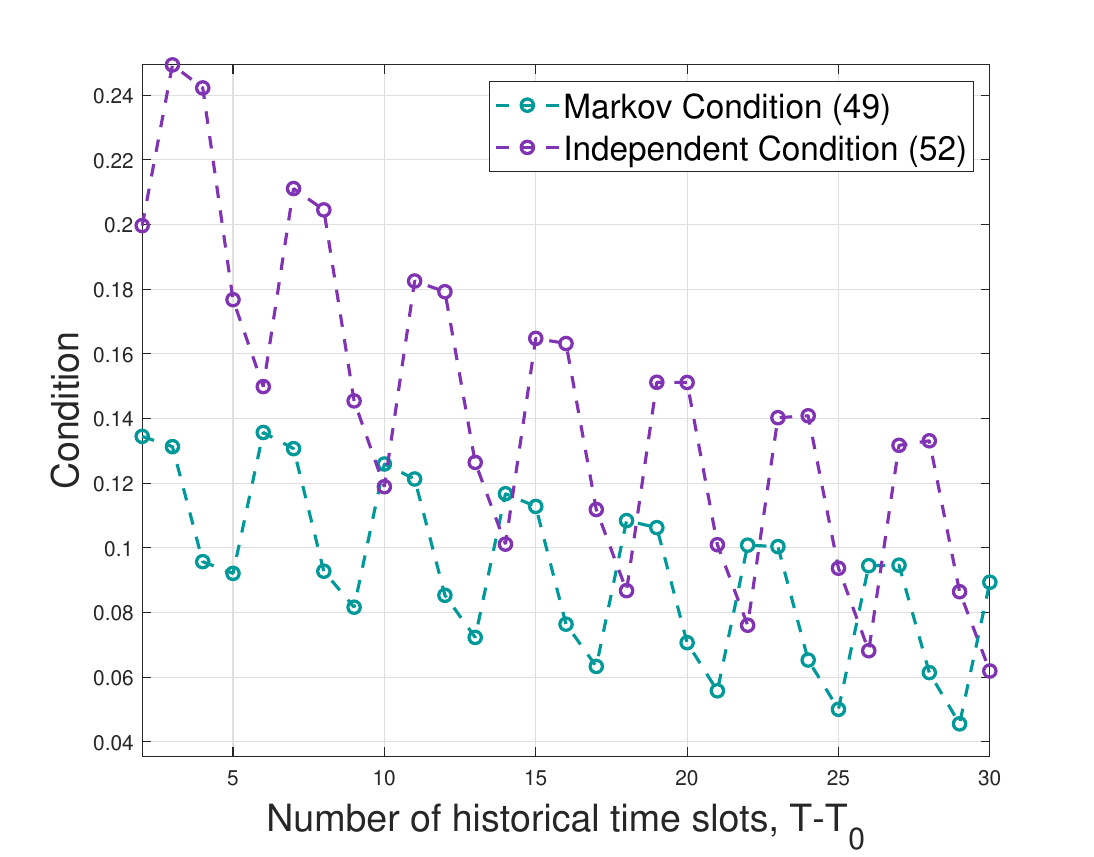}
    \caption{\textcolor{blue}{Markov and independent condition against the historical time range, $T-T_{0}$, when $W=2$, $K=15$, $\tau=2.5\times10^{-4}~\mathrm{s}$ and $\nu = 25~\mathrm{m/s}$.}}
    \label{Fig:Markov_Condition}
\end{figure}

The results in Fig.~\ref{Fig:OP_sampling_ports} and~\ref{Fig:OP_sampling_time_slot} demonstrate that the outage probability of semi-blind FAS can approach that of ideal-FAS, with the cost of sufficient historical port CSI. It also shows that, under appropriate port sampling strategy, e.g, uniform sampling, to achieve the performance as the same as ideal-FAS, the required number of historical CSI can be very limited. We wonder if there excites even better port sampling strategy that can improve the performance of semi-blind FAS further. To this end, we provide the results in Fig.~\ref{Fig:Residual_Entropy_Activited_port}, which investigate the impact of the number of sampling ports on the residual entropy power ratio and outage probability, under different port sampling strategies. Since only a limited set of ports is considered, an exhaustive search method was employed to identify the optimal sampling strategy. As illustrated in Fig.~\ref{Fig:Residual_Entropy_Activited_port}\subref{fig:R_samping_ports}, increasing the number of sampling ports leads to a reduction in the residual entropy power ratio across all three considered methods, with optimal sampling achieving the lowest across all points. Generally, a lower residual entropy power ratio corresponds to better performance. Therefore, the residual entropy power ratio serves as an effective metric for designing the port sampling strategy. However, as shown in Fig.~\ref{Fig:Residual_Entropy_Activited_port}\subref{fig:OP_samping_ports}, the optimal sampling strategy does not consistently achieve the lowest outage probability compared to uniform and sequential sampling. This discrepancy arises because the residual entropy power ratio ensures overall estimation performance, but does not guarantee optimality across every performance metric.

{\color{blue}Finally, Fig. \ref{Fig:Markov_Condition} illustrates the impact of the number of historical time slots on the Markov and independence conditions of the spatial-temporal FAS structure. As the number of historical time slots increases, both conditions tend to be satisfied. The nonlinear decline in both conditions is attributed to the non-monotonic behavior of the correlation function defined in expression \eqref{eq:space_time_correlation_function}. It is worth noting that the Markov condition, i.e., the conditional independence, is more readily satisfied compared to the unconditional independence condition. This suggests that the Markov condition, as outlined in Theorem \ref{theorem:FAS_Markov}, is a more effective criterion for determining the appropriate historical time range for port prediction.}

\section{Conclusion}\label{section:Conclusion}
This paper introduced the spatial-temporal framework of FAS, upon which the semi-blind FAS scheme was developed. \textcolor{black}{The central objective of the semi-blind FAS is to exploit incomplete CSI to approach the performance attainable under complete CSI, that is, the ideal FAS. To this end, semi-blind FAS exploits incomplete historical CSI to estimate the statistical performance of each FAS port at the desired time slot, thereby identifying the optimal port for signal reception.} A closed-form expression for the optimal correlation coefficient, with optimality defined in terms of outage probability minimization, was derived for the special case where only a single historical port CSI is available. \textcolor{black}{Although the optimal port is defined throughout this paper primarily as the one minimizing the outage probability, the semi-blind FAS framework is not limited to this criterion and can be naturally extended to alternative notions of optimality through the same perspective based on the statistical analysis.} To further evaluate semi-blind FAS, we proposed the residual entropy power ratio as a performance metric, and established the Markov condition, which provides useful guidance for practical system design in balancing time slot duration and port switching speed. Numerical results demonstrated that semi-blind FAS achieves performance comparable to, and in some cases indistinguishable from, ideal-FAS, while requiring significantly fewer CSI measurements. The scheme is lightweight, computationally efficient, and scalable to an arbitrary number of ports and time slots, without the need for pre-training or deep learning architectures. These features make semi-blind FAS particularly attractive for deployment on resource-constrained user terminals.

\appendix
\subsection{Proof of Lemma \ref{lemma:Law_Conditional_Normal}}\label{proof:Lemma_Conditional_Normal}
    This result is a direct consequence of the property of conditional distributions of jointly normal random variables. In particular, it is known that when two random variables (or vectors of random variables) are jointly normally distributed, the conditional distribution of one given the other remains normal. 
    Moreover, both the conditional mean and the conditional covariance of this distribution can be expressed explicitly in terms of the means, variances, and covariances of the original variables.  Applying this principle to the specific random variables considered here yields the stated conditional distribution exactly, as detailed in \cite[Proposition 3.13]{conditional_normal}.

\subsection{Proof of Corollary \ref{corollary:OP_given_single_port}}\label{proof:OP_given_single_port}
Conditioned on $h_{\bar{k}}^{T-1} = a_{x}+ja_{y}$, the results in Remark \ref{remark:law_Conditional_Normal_Single} simplified to
\begin{align}\mathcal{L}\left(x_{k}^{T}|x_{{k}^{g}}^{T-1} = a_{x}\right) &= \mathcal{N}\left(\mu_{x,k},\frac{\textcolor{black}{\rho_{k}}}{2}\sigma_{0}^2\right), \\\mathcal{L}\left(y_{k}^{T}|y_{\bar{k}}^{T-1} = a_{y}\right) &= \mathcal{N}\left(\mu_{y,k},\frac{\textcolor{black}{\rho_{k}}}{2}\sigma_{0}^2\right),
\end{align}
where \begin{equation}\label{eq:mean_var_single_port}
    \left\{\begin{aligned}
    \mu_{x,k}&= a_{x}\phi\left(\delta_{k,\bar{k}},\tau\right),\\
    \mu_{y,k}&=  a_{y}\phi\left(\delta_{k,\bar{k}},\tau\right),\\
    \textcolor{black}{\rho_{k}} &= 1-\phi^2\left(\delta_{k,\bar{k}},\tau\right).
    \end{aligned}\right.
    \end{equation}
Subsequently, substituting \eqref{eq:mean_var_single_port} into \eqref{eq:conditional_op}, we achieved \eqref{eq:OP_given_single_port}.

\subsection{Proof of Theorem \ref{theorem:optimal_correlation}}\label{proof:optimal_correlation}
We found the following lemma useful in derivation.
\begin{lemma}
    Let $u = A\sqrt{\frac{\phi^2}{1-\phi^2}}$ and $v = \frac{R}{\sqrt{1-\phi^2}}$, we have the following derivatives as
    \begin{equation}\label{eq:partial_derivative_marcum}
    \left\{\begin{aligned}
   \frac{\partial Q_{1}(u,v)}{\partial u} &= vI_{1}(uv)\exp{\left\{-\frac{u^2+v^2}{2}\right\}},\\
    \frac{\partial Q_{1}(u,v)}{\partial v} &= -vI_{0}(uv)\exp{\left\{-\frac{u^2+v^2}{2}\right\}},
    \end{aligned}\right.
    \end{equation}
    where $I_{0}(\cdot)$ and $I_{1}(\cdot)$ are the modified Bessel function of first kind with order $0$ and $1$, respectively.
\end{lemma}

\begin{proof}
    This proof comes direct from \cite[Eq.(2)\&(5)]{pratt1968partial_marcum}.
\end{proof}

First note that minimum \eqref{eq:OP_given_single_port} is equivalent to maximize $Q_{1}\left(u,v\right)$. Consequently, the first order derivative of $Q_{1}\left(u,v\right)$ with respect to $\phi$ is given by 
\begin{align}
    \frac{\partial Q_{1}\left(u,v\right)}{\partial \phi} &\overset{(a)}{=} \frac{\partial Q_{1}(u,v)}{\partial u}\frac{\partial u}{\partial \phi}+\frac{\partial Q_{1}(u,v)}{\partial v}\frac{\partial v}{\partial \phi}\nonumber \\
    &\overset{(b)}{=} v\exp{\left\{-\frac{u^2+v^2}{2}\right\}}\frac{1}{(1-\phi^2)^{\frac{3}{2}}}\nonumber\\
    &\quad\quad\quad\times\left(AI_{1}(uv)-\phi RI_{0}(uv)\right),
\end{align}
where $(a)$ is the result of the chain rule, and $(b)$ involves the substitution of \eqref{eq:partial_derivative_marcum} with further simplification. Subsequently, the solution of $\frac{\partial Q_{1}\left(u,v\right)}{\partial \phi}=0$ is given by
\begin{align}\label{eq:zero_derivative_marcum}
    g(\phi)=\phi\frac{I_{0}\left(AR\frac{\phi}{1-\phi^2}\right)}{I_{1}\left(AR\frac{\phi}{1-\phi^2}\right)} = \frac{A}{R}.
\end{align}
Noting that $g(-\phi) = g(\phi)$, which implies $\phi = \pm g^{-1}(\frac{A}{R})$.

Moreover, it is evident that the monotonicity of $Q_{1}\left(u,v\right)$ depends solely on $AI_{1}(uv)-\phi RI_{0}(uv)$. However, the monotonicity of $AI_{1}(uv)-\phi RI_{0}(uv)$ is difficult to characterize analytically. Therefore, the solution of \eqref{eq:zero_derivative_marcum} can only be regarded as a candidate optimal point. Consequently, the optimal correlation coefficient of \eqref{eq:OP_given_single_port} is given in \eqref{eq:optimal_correlation}. This ends the proof.

\subsection{Proof of Theorem \ref{theorem:residual_entropy_power_ratio}}\label{proof:residual_entropy_power_ratio}

\textcolor{black}{According to \cite{shannon1948mathematical}} and \cite[Theorem 8.4.1]{cover2005differential}, the entropy power of $\hat{\boldsymbol{\mathrm{h}}}$ is given by
\begin{align}\label{eq:entropy_power_h}
    \mathcal{E}\left(\hat{\boldsymbol{\mathrm{h}}}\right) ={\left(\det{\left(\boldsymbol{\Sigma}_{\hat{\boldsymbol{\mathrm{h}}}} \right)}\right)}^{\frac{1}{K}}.
\end{align}
Consequently, the entropy power of $\hat{\boldsymbol{\mathrm{h}}}{\big|}\bar{\boldsymbol{\mathrm{h}}}$ is given by
\begin{align}\label{eq:entropy_power_conditional_h}
    \mathcal{E}\left(\hat{\boldsymbol{\mathrm{h}}}{\big|}\bar{\boldsymbol{\mathrm{h}}}\right) = {\left(\det{\left(\boldsymbol{\Sigma}_{\hat{\boldsymbol{\mathrm{h}}}|\bar{\boldsymbol{\mathrm{h}}}}\right)}\right)}^{\frac{1}{K}}.
\end{align}
Substituting \eqref{eq:entropy_power_h} and \eqref{eq:entropy_power_conditional_h} into Definition \ref{Def:Residual_entropy_power_ratio}, we have the result that expressed as
\eqref{eq:residual_entropy_power_ratio_partial_FAS}. This ends the proof.

\subsection{Proof of Corollary \ref{corollary:residual_exp_entropy_simplify}}\label{proof:residual_exp_entropy_simplify}
The residual entropy power ratio of the semi-blind FAS can be expressed as
\begin{align}
\mathcal{R}\left(\hat{\boldsymbol{\mathrm{h}}}{\big|}\bar{\boldsymbol{\mathrm{h}}}\right) &\overset{(a)}{=} {\left(\frac{\det{\left(\boldsymbol{\Sigma}_{\hat{\boldsymbol{\mathrm{h}}}} -\boldsymbol{\Sigma}_{\hat{\boldsymbol{\mathrm{h}}},\bar{\boldsymbol{\mathrm{h}}}} \boldsymbol{\Sigma}_{\bar{\boldsymbol{\mathrm{h}}}}^{-1} \boldsymbol{\Sigma}_{\hat{\boldsymbol{\mathrm{h}}},\bar{\boldsymbol{\mathrm{h}}}}^{\top}\right)}}{\det{\left(\boldsymbol{\Sigma}_{\hat{\boldsymbol{\mathrm{h}}}} \right)}}\right)}^{\frac{1}{K}},\\
&\overset{(b)}{=}{\left(\frac{\det{\left(\boldsymbol{\Sigma}_{\hat{\boldsymbol{\mathrm{h}}}}\left(\boldsymbol{\mathrm{I}}_{K} - \boldsymbol{\Sigma}_{\hat{\boldsymbol{\mathrm{h}}}}^{-1}\boldsymbol{\Sigma}_{\hat{\boldsymbol{\mathrm{h}}},\bar{\boldsymbol{\mathrm{h}}}} \boldsymbol{\Sigma}_{\bar{\boldsymbol{\mathrm{h}}}}^{-1} \boldsymbol{\Sigma}_{\hat{\boldsymbol{\mathrm{h}}},\bar{\boldsymbol{\mathrm{h}}}}^{\top}\right)\right)}}{\det{\left(\boldsymbol{\Sigma}_{\hat{\boldsymbol{\mathrm{h}}}} \right)}}\right)}^{\frac{1}{K}},\nonumber
\end{align}
where $(a)$ used the results in \eqref{eq:correlation_Matirx}, \eqref{eq:conditional_correlation_matrix_h} and \eqref{eq:residual_entropy_power_ratio_partial_FAS}, $(b)$ used the fact that $\boldsymbol{\Sigma}_{\hat{\boldsymbol{\mathrm{h}}}}$ is invertible. After fuhrer simplification, we obtain the result that expressed as \eqref{eq:residual_power_entropy_simplify}.

\subsection{Proof of Theorem \ref{theorem:FAS_Markov}}\label{proof:FAS_Markov}
Taking into account the Markov process, the correlation matrix of $\left[\boldsymbol{\mathrm{h}}_{\mathcal{\hat{K}}^{T}}^{T} ,\boldsymbol{\mathrm{h}}_{\mathcal{\bar{K}}^{T_{0}-1}}^{T_{0}-1},\boldsymbol{\mathrm{h}}_{\mathcal{\bar{K}}^{T_{0}}}^{T_{0}}\right]$
is given by 
\begin{align}
    \left[\begin{array}{ccc}
    \boldsymbol{\Sigma}_{{\mathcal{\hat{K}}^{T}}} &\boldsymbol{\Sigma}_{{\mathcal{\hat{K}}^{T}},{\mathcal{\bar{K}}^{T_{0}-1}}} &\boldsymbol{\Sigma}_{{\mathcal{\hat{K}}^{T}},{\mathcal{\bar{K}}^{T_{0}}}}\\
     \boldsymbol{\Sigma}_{{\mathcal{\hat{K}}^{T}},{\mathcal{\bar{K}}^{T_{0}-1}}}^{\top} & \boldsymbol{\Sigma}_{{\mathcal{\bar{K}}^{T_{0}-1}}} &\boldsymbol{\Sigma}_{{\mathcal{\bar{K}}^{T_{0}-1}},{\mathcal{\bar{K}}^{T_{0}}}}\\ \boldsymbol{\Sigma}_{{\mathcal{\hat{K}}^{T}},{\mathcal{\bar{K}}^{T_{0}}}}^{\top}& \boldsymbol{\Sigma}_{{\mathcal{\bar{K}}^{T_{0}-1}},{\mathcal{\bar{K}}^{T_{0}}}}^{\top}& \boldsymbol{\Sigma}_{{\mathcal{\bar{K}}^{T_{0}}}}
    \end{array}
    \right],
\end{align}which can be further expressed as
\begin{align}\label{eq:space_time_cov_matrix_Markov}\left[
    \begin{array}{cc}\boldsymbol{\Sigma}_{{\mathcal{\hat{K}}^{T}},{\mathcal{\bar{K}}^{T_{0}-1}}}   &\boldsymbol{\Sigma}_{({\mathcal{\hat{K}}^{T}},{\mathcal{\bar{K}}^{T_{0}-1}}),\mathcal{\bar{K}}^{T_{0}}}\\ \boldsymbol{\Sigma}_{({\mathcal{\hat{K}}^{T}},{\mathcal{\bar{K}}^{T_{0}-1}}),\mathcal{\bar{K}}^{T_{0}}}^{\top}& \boldsymbol{\Sigma}_{{\mathcal{\bar{K}}^{T_{0}}}}
    \end{array}
    \right],
\end{align}
where each sub-matrix is given in \eqref{eq:M_matrix_component}. The conditional correlation matrix between $\boldsymbol{\mathrm{h}}_{\mathcal{\hat{K}}^{T}}^{T}$and $\boldsymbol{\mathrm{h}}_{\mathcal{\bar{K}}^{T_{0}-1}}^{T_{0}-1}$ given $\boldsymbol{\mathrm{h}}_{\mathcal{\bar{K}}^{T_{0}}}^{T_{0}}$, is given by \eqref{eq:conditional_correlation_matrix_Markov}. Consequently, conditioned on $\boldsymbol{\mathrm{h}}_{\mathcal{\bar{K}}^{T_{0}}}^{T_{0}}$, the cross-correlation between $\boldsymbol{\mathrm{h}}_{\mathcal{\hat{K}}^{T}}^{T}$ and $\boldsymbol{\mathrm{h}}_{\mathcal{\bar{K}}^{T_{0}-1}}^{T_{0}-1}$ is captured by the submatrix of the conditional correlation matrix in \eqref{eq:conditional_correlation_matrix_Markov}, specifically, given by the block $(\boldsymbol{\Sigma}_{{\mathcal{\hat{K}}^{T}}, {\mathcal{\bar{K}}^{T_{0}-1}} \mid \mathcal{\bar{K}}^{T_{0}}})_{1:K,1+K:2K}$. Subsequently, according to \cite[Definition 4.9, Theorem 4.1]{Markov}, the Markov conditional independence property is satisfied given 
\begin{align}
    (\boldsymbol{\Sigma}_{{\mathcal{\hat{K}}^{T}}, {\mathcal{\bar{K}}^{T_{0}-1}} \mid \mathcal{\bar{K}}^{T_{0}}})_{1:K,1+K:2K} = \boldsymbol{\mathrm{0}},
\end{align}
where $\boldsymbol{\mathrm{0}}$ is the zero matrix. However, the zero matrix is hard to achieve due to the compactness structure of FAS, and therefore, we approximated it as expressed in \eqref{eq:FAS_Markov_condition}, where $K^2$ serves as the normalization factor.

\bibliographystyle{IEEEtran}
\bibliography{Reference_List.bib}

\end{document}